\documentclass[11pt]{article}
\usepackage[margin=1in]{geometry}
\usepackage[colorlinks=true,
            linkcolor=blue,
            citecolor=blue,
            urlcolor=blue]{hyperref}

\usepackage{fixmath}
\usepackage{bm}
\usepackage{amsbsy}
\usepackage{color}
\usepackage{verbatim}
\usepackage{multirow}
\usepackage{amssymb}
\usepackage{amsthm}
\usepackage{array}
\usepackage{mathtools}
\usepackage{amsmath}
\usepackage{graphicx}
\usepackage{tikz}
\usetikzlibrary{positioning}
\usepackage{xcolor}

\usepackage{booktabs}
\usepackage{subcaption}
\usepackage{thm-restate}

\newtheorem{theorem}{Theorem}
\newtheorem{lemma}{Lemma}

\newtheorem{assumption}{Assumption}

\newtheorem{remark}{Remark}

\usepackage{algorithm}
\usepackage{algorithmicx}
\usepackage{algpseudocode}

\usepackage{authblk}

\usepackage{listings}
\definecolor{codegreen}{rgb}{0,0.6,0}
\definecolor{codegray}{rgb}{0.5,0.5,0.5}
\definecolor{codepurple}{rgb}{0.58,0,0.82}
\definecolor{backcolour}{rgb}{0.95,0.95,0.92}

\usepackage{acro}

\newcommand{\size}[1]{\ensuremath{|#1|}}

\newcommand{\lrA}[1]{\ensuremath{\left(#1\right)}}

\def\B{\mathcal{B}}

\def\OPT{\mbox{OPT}}

\allowdisplaybreaks

\title{An Improved Approximation Algorithm for Maximum Weight 3-Path Packing}

\author[1,2]{Jingyang Zhao} 
\author[1]{Mingyu Xiao}

\affil[1]{University of Electronic Science and Technology of China}
\affil[2]{Kyung Hee University, Yongin-si, South Korea}

\date{}

\begin{document}

\maketitle

\begin{abstract}
Given a complete graph with $n$ vertices and non-negative edge weights, where $n$ is divisible by 3, the maximum weight 3-path packing problem is to find a set of $n/3$ vertex-disjoint 3-paths such that the total weight of the 3-paths in the packing is maximized. This problem is closely related to the classic maximum weight matching problem. In this paper, we propose a $10/17$-approximation algorithm, improving the best-known $7/12$-approximation algorithm (ESA 2015). Our result is obtained by making a trade-off among three algorithms. The first is based on the maximum weight matching of size $n/2$, the second is based on the maximum weight matching of size $n/3$, and the last is based on an approximation algorithm for star packing. Our first algorithm is the same as the previous $7/12$-approximation algorithm, but we propose a new analysis method---a charging method---for this problem, which is not only essential to analyze our second algorithm but also may be extended to analyze algorithms for some related problems.
\end{abstract}

\maketitle

\section{Introduction}
A \textit{3-path} in a graph is a simple path on three different vertices. 
A \emph{3-path packing} in a graph is a set of vertex-disjoint 3-paths.
Given an unweighted graph, the Maximum 3-Path Packing (M3PP) problem is to find a 3-path packing of maximum cardinality.
For an edge-weighted complete graph, 
we assume that every edge has a non-negative weight and the number of vertices in the graph $n$ is a multiple of 3. 
A 3-path packing is \emph{perfect} if its size is $n/3$ (i.e., it can cover all vertices).
The Maximum Weight 3-Path Packing (MW3PP) problem is to find a perfect 3-path packing in an edge-weighted complete graph such that the total weight of all 3-paths is maximized.
MW3PP is a natural special case of the team formation problem~\cite{baumer2015star,DBLP:journals/dam/Bar-NoyPRV18}, where the goal is to find a set of teams, each containing a leader and two other members, such that the total weight of the connections between the team leader and the other team members is maximized.

In this paper, we study approximation algorithms for MW3PP.

\subsection{Related work}
It is well-known that deciding whether an unweighted graph contains a perfect 3-path packing is NP-hard~\cite{KirkpatrickH78}. Hence, it is easy to see that both M3PP and MW3PP are NP-hard.

MW3PP is a special case of the weighted 3-Set Packing problem, which admits approximations ratio of $(1/2-\varepsilon)$~\cite{arkin1998local,DBLP:journals/njc/Berman00}, $[1/(2-1/63700992+\varepsilon)\approx 0.5+10^{-9}]$~\cite{DBLP:conf/stacs/Neuwohner21}, $(1/1.786\approx 0.5599)$~\cite{thiery2023improved}, and $(1/\sqrt{3}\approx 0.5773)$~\cite{thiery2023approximation}. For MW3PP, there are independent approximation algorithms. 
Specifically, Hassin and Rubinstein~\cite{hassin2006approximation} proposed a randomized $(0.5223-\varepsilon)$-approximation algorithm. Tanahashi and Chen~\cite{tanahashi2010deterministic} proposed a deterministic $(0.5265-\varepsilon)$-approximation algorithm. 
Both the previous two algorithms are based on a maximum weight cycle packing~\cite{hartvigsen1984extensions}.
Bar-Noy \emph{et al.}~\cite{DBLP:journals/dam/Bar-NoyPRV181,DBLP:journals/dam/Bar-NoyPRV18} proposed an improved $7/12$-approximation algorithm, which is based on a maximum weight matching of size $n/2$~\cite{gabow1974implementation,lawler1976combinatorial}.
For MW3PP on $\{0,1\}$-weighted graphs (i.e., a complete graph with edge weights of 0 and 1), Hassin and Schneider~\cite{hassin2013local} gave a $0.55$-approximation algorithm, which was later improved to $3/4$ by Bar-Noy \emph{et al.}~\cite{DBLP:journals/dam/Bar-NoyPRV181,DBLP:journals/dam/Bar-NoyPRV18}.

\subsection{Our results}
In this paper, we improve the approximation ratio from $7/12\approx 0.58333$ to $10/17>0.58823$ for MW3PP. Our ratio is achieved by making a trade-off among three algorithms.

Let $P^*$ be an optimal 3-path packing with the weight of $\OPT$. 
At the high-level, our first algorithm is to call the previous $7/12$-approximation algorithm~\cite{DBLP:journals/dam/Bar-NoyPRV18}, which is based on a maximum weight matching of size $n/2$. The framework of the previous analysis of this algorithm is based on so-called uniformly $c$-sparse graphs~\cite{DBLP:journals/dam/Bar-NoyPRV18}, which is a global graph structure property. However, we will give a counterexample to show that their analysis is incomplete, and hence we cannot directly employ their result. For the sake of completeness, we provide an alternative proof based on several local graph structure properties. Our new analysis method is also essential to analyze our second algorithm.
Our second algorithm is based on a maximum weight matching of size $n/3$. It is motivated by the following observation: in the worst case of the first algorithm the weight of the maximum weight matching of size $n/2$ is the same as the weight of the maximum weight matching of size $n/3$. 
A combination of the first two algorithms still cannot guarantee a better-than-$7/12$-approximation ratio. However, in the worst case of the first two algorithms we get that all of the weighted parts of the fixed optimal 3-path packing $P^*$ lie within the graph $LG$ induced by the vertex set of the maximum weight matching of size $n/3$. This special case is handled by our third algorithm, which is based on a good 2-star packing (i.e., a set of vertex-disjoint edges $K_{1,1}$ and 3-paths $K_{1,2}$ covering all vertices) in $LG$. Based on the $2/3$-approximation algorithm for the maximum weight 2-star packing problem~\cite{babenko2011new}, we show that we can obtain a good 2-star packing in $LG$ with a weight of at least $\frac{2}{3}\cdot\OPT>\frac{7}{12}\cdot\OPT$ in the special case. Since there are $n/3$ vertices in $RG\coloneqq G-LG$, any 2-star packing in $LG$ can be transformed into a 3-path packing of $G$ with a non-decreasing weight. Hence, we can obtain a 3-path packing with a weight of at least $\frac{2}{3}\cdot\OPT$ in the special case. At last, by making a trade-off among these three algorithms (i.e., returning the best solution among the results of these three algorithms), we show that we can obtain a better-than-$7/12$-approximation ratio.

\subsection{Further related work}
A \textit{$k$-path/cycle} in a graph is a simple path/cycle on $k$ different vertices. 
Analogously, we denote by M$k$PP and M$k$CP the Maximum $k$-Path Packing and Maximum $k$-Cycle Packing problems, respectively; and by MW$k$PP and MW$k$CP the Maximum Weight $k$-Path Packing and Maximum Weight $k$-Cycle Packing problems, respectively.

M$k$PP and M$k$CP are NP-hard for $k\geq 3$~\cite{KirkpatrickH78}; thus, MW$k$PP and MW$k$CP are NP-hard for $k\geq 3$, even on metric graphs where the weight function is a metric.
Moreover, M$k$PP and M$k$CP are APX-hard for $k\geq 3$~\cite{DBLP:journals/ipl/Kann94}; MW$k$PP and MW$k$CP are APX-hard on $\{0,1\}$-weighted graphs for $k\geq 3$~\cite{manthey2008approximating}.
Guruswam et al.~\cite{DBLP:conf/wg/GuruswamiRCCW98} showed that M3CP remains NP-hard on chordal graphs, planar graphs, line graphs, and total graphs.

Approximation algorithms for MW$k$PP and MW$k$CP have been studied extensively on both general and metric graphs.
Notably, MW$n$CP is known as MAX TSP.
In this subsection, we use $\alpha$ (resp., $\beta$) to denote the best-known approximation ratio of MAX TSP on general (resp., metric) graphs. 
Currently, $\alpha=4/5$~\cite{DBLP:conf/ipco/DudyczMPR17} and $\beta=7/8$~\cite{DBLP:journals/tcs/KowalikM09}.

\textbf{MW$k$PP.}
Hassin and Rubinstein~\cite{hassin1997approximation} proposed a $3/4$-approximation algorithm for $k=4$, and an $\alpha\cdot(1-1/k)$-approximation algorithm for $k\geq5$.
On $\{0,1\}$-weighted graphs, Berman and Karpinski~\cite{DBLP:conf/soda/BermanK06} gave a $6/7$-approximation algorithm for $k=n$.

\textbf{MW$k$CP.}
For $k=3$, Hassin and Rubinstein~\cite{hassin2006approximation,DBLP:journals/dam/HassinR06a} proposed a randomized $(0.518-\varepsilon)$-approximation algorithm. Chen et al.~\cite{DBLP:journals/dam/ChenTW09,DBLP:journals/dam/ChenTW10} improved this with a randomized $(0.523-\varepsilon)$-approximation algorithm. 
Van Zuylen~\cite{DBLP:journals/dam/Zuylen13} subsequently proposed a deterministic algorithm achieving the same approximation ratio. 
Moreover, Li and Yu~\cite{li2023cyclepack} gave a $2/3$-approximation algorithm for $k=4$ and an $\alpha\cdot(1-1/k)^2$-approximation algorithm for $k\geq5$.
Zhao and Xiao~\cite{DBLP:journals/tcs/ZhaoX25} gave an improved $3/4$-approximation algorithm for $k=4$.
On $\{0,1\}$-weighted graphs, Bar-Noy \emph{et al.}~\cite{DBLP:journals/dam/Bar-NoyPRV181,DBLP:journals/dam/Bar-NoyPRV18} gave a $3/5$-approximation algorithm for $k=3$, and the algorithm of Berman and Karpinski~\cite{DBLP:conf/soda/BermanK06} yields an approximation ratio of $6/7$ for $k=n$.

\textbf{MW$k$PP on metric graphs.}
Li and Yu~\cite{li2023cyclepack} gave a $3/4$-approximation algorithm for $k=3$, a $3/4$-approximation algorithm for $k=5$, and a $\beta\cdot(1-1/k)$-approximation algorithm for $k\geq6$. 
Zhao and Xiao~\cite{DBLP:journals/tcs/ZhaoX25} gave an improved $14/17$-approximation algorithm for $k=4$, and an improved $\frac{27k^2-48k+16}{32k^2-36k-24}$-approximation algorithm for $k\in\{6,8,10\}$.
On $\{1,2\}$-weighted graphs, Monnot and Toulouse~\cite{DBLP:journals/jda/MonnotT08} gave a $9/10$-approximation algorithm for $k=4$, and Guo et al.~\cite{DBLP:journals/jco/GuoYL25} gave a $19/24$-approximation algorithm for $k=5$. 
Notably, Zhao and Xiao~\cite{DBLP:journals/tcs/ZhaoX25} showed that for MW$k$PP (and MW$k$CP), any $\rho$-approximation algorithm on $\{0,1\}$-weighted graphs yields a $(1+\rho)/2$-approximation algorithm on $\{1,2\}$-weighted graphs. 
Then, since there exists a $24/35$-approximation algorithm for MW$5$PP on $\{0,1\}$-weighted graphs~\cite{DBLP:conf/soda/BermanK06,li2023cyclepack}, one can obtain a ratio of $59/70$ for MW$5$PP on $\{1,2\}$-weighted graphs, better than $19/24$~\cite{DBLP:journals/jco/GuoYL25}.

\textbf{MW$k$CP on metric graphs.}
For $k=3$, Hassin et al.~\cite{hassin1997approximation1} proposed a $2/3$-approximation algorithm, Chen et al.~\cite{DBLP:journals/jco/ChenCLWZ21} gave a randomized $(0.66768-\varepsilon)$-approximation algorithm, and Zhao and Xiao~\cite{DBLP:journals/tcs/ZhaoX24} gave a deterministic $(0.66835-\varepsilon)$-approximation algorithm.
Moreover, Li and Yu~\cite{li2023cyclepack} gave a $3/4$-approximation algorithm for $k=4$, a $3/5$-approximation algorithm for $k=5$, and a $\beta\cdot(1-1/k)^2$-approximation algorithm for $k\geq6$. 
Zhao and Xiao~\cite{DBLP:journals/tcs/ZhaoX25} improved these ratios to $5/6$ for $k=4$, to $7/10$ for $k=5$, to $(\frac{7}{8}-\frac{1}{8k})(1-\frac{1}{k})$ for fixed odd $k>5$, and to $\frac{7}{8} (1-\frac{1}{k}+\frac{1}{k(k-1)})$ for even $k>5$.
On $\{1,2\}$-weighted graphs, Zhao and Xiao~\cite{DBLP:journals/tcs/ZhaoX25} proposed a $9/11$-approximation algorithm for $k=3$ and a $7/8$-approximation algorithm for $k=4$, and Guo et al.~\cite{DBLP:journals/jco/GuoYL25} gave a $37/48$-approximation algorithm for $k=6$. Similarly, 
since there exists a $25/42$-approximation algorithm for MW$6$CP on $\{0,1\}$-weighted graphs~\cite{DBLP:conf/soda/BermanK06,li2023cyclepack}, one can obtain a ratio of $67/84$ for MW$6$CP on $\{1,2\}$-weighted graphs~\cite{DBLP:journals/tcs/ZhaoX25}, better than $37/48$~\cite{DBLP:journals/jco/GuoYL25}.

Finally, MW$k$PP and MW$k$CP (with fixed $k$) are special cases of weighted $k$-Set Packing, which admits approximation ratios of $\frac{1}{k-1}-\varepsilon$~\cite{arkin1998local}, $\frac{2}{k+1}-\varepsilon$~\cite{DBLP:journals/njc/Berman00}, and $\frac{2}{k+1-1/31850496}-\varepsilon$~\cite{DBLP:conf/stacs/Neuwohner21}. Recently, these results have been further improved (see~\cite{DBLP:conf/ipco/Neuwohner22,DBLP:conf/soda/Neuwohner23,thiery2023improved,thiery2023approximation}). 
To our knowledge, the only case where these improvements yield the state-of-the-art for MW$k$PP/MW$k$CP is MW3CP, namely the $1/\sqrt{3}$-approximation algorithm for weighted 3-Set Packing~\cite{thiery2023approximation}.

For additional results on the minimization versions, we refer the reader to~\cite{DBLP:journals/siamcomp/GoemansW95,DBLP:journals/jda/MonnotT08,DBLP:conf/aaim/LiYL24,DBLP:conf/cocoon/XuYL25}.

\section{Preliminaries}\label{prelim}
An instance of MW3PP is given by a complete graph $G=(V,E)$ with $n$ vertices, where $n\bmod 3=0$. There is a non-negative weight function $w: E\to \mathbb{R}_{\geq0}$ on the edges in $E$.
For any weight function $w:X\to \mathbb{R}_{\geq0}$, we also extend it to the subsets of $X$, i.e., we define $w(X') = \sum_{x\in X'} w(x)$ for any $X'\subseteq X$. Given a subgraph $S$, we use $V(S)$ (resp., $E(S)$) to denote the set of vertices (resp., edges) contained in the subgraph, and for any function $f$ on the edges, define $f(S)\coloneqq f(E(S))$. Given a set of vertices $V'\subseteq V$, we use $G[V']$ to denote the complete graph induced by $V'$. 
Two subgraphs or sets of edges are \emph{vertex-disjoint} if they do not share a common vertex. 
An undirected edge between vertices $u$ and $v$ is denoted by $uv$, and a directed edge from $u$ to $v$ is denoted by $(u,v)$. A directed edge will be called an \emph{arc}.
We may consider a directed complete graph $\overleftrightarrow{G}=(V,A)$ obtained by replacing each edge $uv\in E$ of $G$ by two arcs $\{(u,v),(v,u)\}$.

A \emph{matching} is a set of vertex-disjoint edges. We use $M^*_{p}$ to denote a maximum weight matching of size $p$ in graph $G$, which can be found in $O(n^3)$ time if it exists~\cite{gabow1974implementation,lawler1976combinatorial}. The maximum weight matching without the constraint on its size can also be found in $O(n^3)$ time~\cite{gabow1974implementation,lawler1976combinatorial}.
We assume that $n$ is even so that $M^*_{n/2}$ exists; otherwise, we use $O(n^2)$ loops to enumerate one 3-path in the optimal solution of MW3PP, and in the loop we get an approximate 3-path packing in the graph obtained by removing the 3-path, where the number of vertices is even. The approximation ratio is preserved at a cost of a runtime increase by a factor of $O(n^2)$.

A 3-path $xyz$ is a simple path on three different vertices $\{x,y,z\}$, and it contains exactly two edges $\{xy,yz\}$.
A \emph{perfect 3-path packing} in graph $G$ is a set of $n/3$ vertex-disjoint 3-paths that cover all vertices of $G$.
In the following, we will always consider a 3-path packing as a perfect 3-path packing and use $P^*$ to denote a maximum weight 3-path packing in $G$. We also let $\OPT\coloneqq w(P^*)$, and assume $\OPT\neq 0$ to exclude this trivial case.

A $t$-star is a complete bipartite graph $K_{1,t}$, where there is one vertex on one side, $t$ vertices on the other side, and $t$ edges in total. Hence, an edge is a 1-star, and a 3-path is a 2-star. A \emph{$t$-star packing} in a graph is a set of vertex-disjoint $t'$-stars with $1\leq t'\leq t$ that cover all vertices of the graph. Note that a 2-star packing is a set of edges and 3-paths, and a 3-path packing is also a 2-star packing by definition. 
To avoid confusion with normal edges, we refer to the edges in a star packing as a \emph{2-path}.
In a directed graph, the in-degree (resp., out-degree) of a vertex $v$ is the number of arcs entering (resp., leaving) $v$. A \emph{$t$-feasible arc set} in a directed graph is a subgraph such that each vertex has an in-degree of at most 1 and an out-degree of at most $t$.

Due to limited space, the proofs of the lemmas and theorems marked with ``*'' have been moved to the appendix.

\section{The First Algorithm}
The first algorithm has already
been considered in~\cite{DBLP:journals/dam/Bar-NoyPRV18}, and we denote it by Alg.1. It contains the following four steps. 

\medskip
\noindent\textbf{Step~1.} Find $M^*_{n/2}$, a maximum weight matching of size $n/2$ in the graph $G$.

\noindent\textbf{Step~2.} Construct a multi-graph $G/M^*_{n/2}$ by contracting the edges
in $M^*_{n/2}$, where there are $n/2$ super-vertices one-to-one corresponding to the edges in $M^*_{n/2}$. For two super-vertices corresponding to $ux,yz\in M^*_{n/2}$, there are four edges $uy,uz,xy,xz$ between them, and each edge $e\in\{uy,uz,xy,xz\}$ is defined to have a cost of $c(e)\coloneqq w(e)-\min\{w(ux),w(yz)\}$.

\noindent\textbf{Step~3.} Find a maximum cost matching $M^{**}_{n/6}$ of size $n/6$ in the graph $G/M^*_{n/2}$.

\noindent\textbf{Step~4.} Obtain a 3-path packing $P_1$ of $G$ as follows:
\begin{itemize}
    \item For each edge $xy\in M^{**}_{n/6}$, if there are two edges $e_x,e_y\in M^*_{n/2}$ corresponding to the endpoints of $xy$ in $G/M^*_{n/2}$, where we assume that $w(e_x)\geq w(e_y)$, obtain a 3-path using edges $xy$ and $e_x$, and call the vertex of $e_y$ that is not incident to $xy$ a \emph{residual vertex}.
    \item For each edge $e\in M^*_{n/2}$ corresponding to a vertex of $G/M^*_{n/2}$ that is not matched by $M^{**}_{n/6}$, create a 3-path from $e$ with an arbitrary residual vertex.
\end{itemize}
\medskip

An illustration of Alg.1 can be seen in Fig.\ref{fig:fig01}.
\begin{figure}[t]
    \centering
    \includegraphics[scale=0.63]{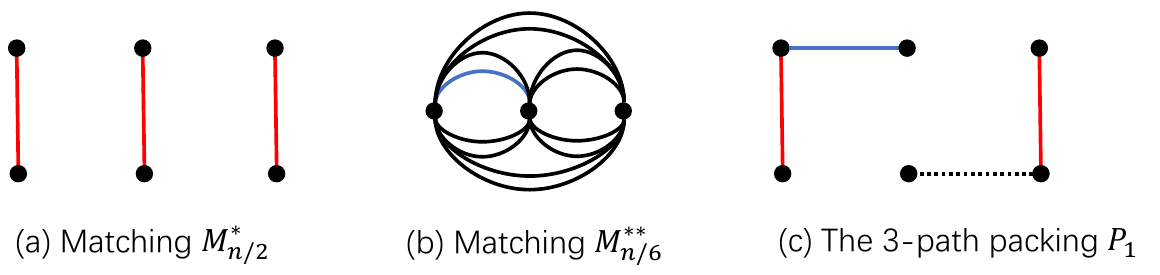}
    \caption{An illustration of Alg.1: In (a), each red edge has a weight of 1, each omitted edge has a weight of 0, and $M^*_{n/2}$ contains three red edges; In (b), each edge has a cost of -1 by definition and $M^{**}_{n/6}$ contains one blue edge; In (c), Alg.1 outputs two 3-paths with a weight of $2$.}
    \label{fig:fig01}
\end{figure}

\subsection{A note on the previous analysis}
An easy bound related to $M^*_{n/2}$ can be obtained as follows.
\begin{lemma}\label{lb1}
We have $w(M^*_{n/2})\geq w(M^*_{n/3})\geq\sum_{xyz\in P^*}\max\{w(xy),w(yz)\}\geq \frac{1}{2}\cdot\OPT$.
\end{lemma}
\begin{proof}
Since the weights of edges are non-negative, we have $w(M^*_{n/2})\geq w(M^*_{n/3})$.
Since the optimal 3-path packing $P^*$ contains $n/3$ 3-paths, it is easy to see that $\{xy\mid xyz\in P^*,w(xy)\geq w(yz)\}$ is a matching of size $n/3$. Hence, we have $w(M^*_{n/3})\geq \sum_{xyz\in P^*}\max\{w(xy),w(yz)\}\geq\sum_{xyz\in P^*}\frac{1}{2}(w(xy)+w(yz))=\sum_{xyz\in P^*}\frac{1}{2}w(xyz)=\frac{1}{2}w(P^*)=\frac{1}{2}\cdot\OPT$.
\end{proof}

Bar-Noy \emph{et al.}~\cite{DBLP:journals/dam/Bar-NoyPRV18} proved that the approximation of Alg.1 is $\frac{7}{12}$ using the following result.

\begin{lemma}[\cite{DBLP:journals/dam/Bar-NoyPRV18}, cf. theorem~1]\label{lb2}
It holds that $w(P_1)\geq \frac{2}{3}w(M^*_{n/2})+\frac{1}{4}\cdot\OPT$.
\end{lemma}

By Lemmas~\ref{lb1} and \ref{lb2}, it is easy to get that $w(P_1)\geq \frac{7}{12}\cdot\OPT$. However, Lemma~\ref{lb2} may not hold. In the example shown in Fig.\ref{fig:fig01}, we have $w(M^*_{n/2})=3$, $\OPT=2$, and $w(P_1)=2$, while we get $w(P_1)\geq 2+\frac{1}{2}>\OPT$ by Lemma~\ref{lb2}. Hence, Lemma~\ref{lb2} implies that the weight of $P_1$ is even greater than the optimal solution, a contradiction. 

Note that the costs of the edges in $G/M^*_{n/2}$ may be negative. Due to this, some inequalities in the proof of Lemma~\ref{lb2} in~\cite{DBLP:journals/dam/Bar-NoyPRV18} may not hold.
By additionally analyzing this case using their method, we may still prove an approximation ratio of $\frac{7}{12}$ for Alg.1. However, we provide a new analysis method since it is also essential to analyze our second algorithm.

\subsection{The analysis}
First, we consider the quality of the 3-path packing $P_1$ returned by Alg.1.
\begin{lemma}[\cite{DBLP:journals/dam/Bar-NoyPRV18}]\label{lb3}
It holds that $w(P_1)\geq w(M^*_{n/2})+c(M^{**}_{n/6})$.
\end{lemma}
\begin{proof}
Let $M'$ be the set of edges in $M^*_{n/2}$ such that each edge corresponds to a vertex of $G/M^*_{n/2}$ that is matched by $M^{**}_{n/6}$, and let $M''\coloneqq M^*_{n/2}\setminus M'$. Note that $\size{M'}=n/3$ and $\size{M''}=\size{M^*_{n/2}}-\size{M'}=n/2-n/3=n/6$.

By Step 4 of Alg.1, we have that the 3-paths in $P_1$ contain all edges in $M^{**}_{n/6}$, $n/6$ edges in $M'$, and all edges in $M''$ (see the example in Fig.\ref{fig:fig01}).
Specifically, for each edge $xy\in M^{**}_{n/6}$, there are two edges $e_x,e_y\in M^*_{n/2}$ corresponding to the endpoints of $xy$ in $G/M^*_{n/2}$, and we have $e_x,e_y\in M'$. If $w(e_x)\geq w(e_y)$, there is a 3-path containing two edges $xy$ and $e_x$, and we have $w(xy)+w(e_x)=w(xy)+w(e_x)+w(e_y)-\min\{w(e_x),w(e_y)\}=w(e_x)+w(e_y)+c(xy)$. Hence, the weight of the $n/6$ 3-paths using $n/6$ edges in $M^{**}_{n/6}$ and $n/6$ edges in $M'$ has a weight of $w(M')+c(M^{**}_{n/6})$. Moreover, for each $e\in M''$, there exists a 3-path in $P_1$ containing it with a residual vertex. Since $\size{M''}=n/6$, the weight of these $n/6$ 3-paths has a weight of at least $w(M'')$. Hence, we have $w(P_1)\geq w(M')+c(M^{**}_{n/6})+w(M'')=w(M^*_{n/2})+c(M^{**}_{n/6})$.
\end{proof}

We have obtained a bound related to $w(M^*_{n/2})$ in Lemma~\ref{lb1}.
To obtain a bound related to $c(M^{**}_{n/6})$, we first define some notation. We split 3-paths in $P^*$ into two disjoint sets $P^*_1$ and $P^*_2$, where $P^*_1\coloneqq\{xyz\mid xyz\in P^*,\{xy,yz\}\cap E(M^*_{n/2})=\emptyset\}$ and $P^*_2\coloneqq P^*\setminus P^*_1$. Hence, $P^*_1$ contains all 3-paths in $P^*$ with edges not contained in $M^*_{n/2}$, and for each 3-path in $P^*_2$, there is exactly one edge contained in $M^*_{n/2}$. In this section, we let 
\[
E_1\coloneqq E(P^*_1)\quad\quad\mbox{and}\quad\quad E_2\coloneqq E(P^*_2)\setminus E(M^*_{n/2}).
\]

\begin{lemma}\label{lb4}
For any edge $xy$ in $G/M^*_{n/2}$, there are two edges $e_x,e_y\in M^*_{n/2}$ corresponding to the endpoints of $xy$ in $G/M^*_{n/2}$. %Assume that $e_x\neq e_y$. 
For any $0\leq \theta\leq 1$, it holds that $c(xy)\geq w(xy)-(\theta\cdot w(e_x)+(1-\theta)\cdot w(e_y))$.
\end{lemma}
\begin{proof}
We have $c(xy)=w(xy)-\min\{w(e_x), w(e_y)\}$ by definition. Assume w.l.o.g. that $w(e_x) = \min\{w(e_x), w(e_y)\}\leq w(e_y)$. We get $c(xy)=w(xy)-w(e_x)=w(xy)-(\theta\cdot w(e_x)+(1-\theta)\cdot w(e_x))\geq w(xy)-(\theta\cdot w(e_x)+(1-\theta)\cdot w(e_y))$ since $0\leq \theta\leq 1$.
\end{proof}

Based on the result in Lemma~\ref{lb4}, we use a charging method (a new analysis method for this problem) to prove the following lemma.
\begin{lemma}\label{lb5}
It holds that $c(E(P^*)\setminus E(M^*_{n/2}))\geq w(P^*)-\frac{4}{3}w(M^*_{n/2})$.
\end{lemma}
\begin{proof}
By definition, we have $E(P^*)\setminus E(M^*_{n/2}) = E_1\cup E_2$ and $E(P^*)\cap E(M^*_{n/2}) = E(P^*_2)\cap E(M^*_{n/2})$. Hence, $E(P^*) = E_1\cup E_2\cup (E(P^*_2)\cap E(M^*_{n/2}))$, and it suffices to prove
\[
c(E_1)+c(E_2)\geq w(E_1) + w(E_2) + w(E(P^*_2)\cap E(M^*_{n/2})) - \frac{4}{3}w(M^*_{n/2}).
\]

For any edge $xy\in E_1\cup E_2$, there are two edges $e_x,e_y\in M^*_{n/2}$ corresponding to the endpoints of $xy$ in $G/M^*_{n/2}$. By Lemma~\ref{lb4}, it holds that $c(xy)\geq w(xy)-(\theta\cdot w(e_x)+(1-\theta)\cdot w(e_y))$ for any $0\leq \theta\leq 1$. We say that $xy$ charges $\theta$ point from $e_x$ and $1-\theta$ point from $e_y$. Next, we design a charging rule so that each edge in $E_1\cup E_2$ receives 1 point in total.

A vertex is called an \emph{$i$-degree vertex} if it is incident to $i$ edges in $E(P^*)$. For any edge $xy\in E_1\cup E_2$, there exist exactly one 2-degree vertex and one 1-degree vertex in $\{x,y\}$. Assume w.l.o.g. that $x$ is a 1-degree vertex and $y$ is a 2-degree vertex. We let $xy$ charge $\frac{2}{3}$ point from $e_x$ and $\frac{1}{3}$ point from $e_y$. See Fig.\ref{fig:fig02} for an illustration.
\begin{figure}[t]
    \centering
    \includegraphics[scale=0.58]{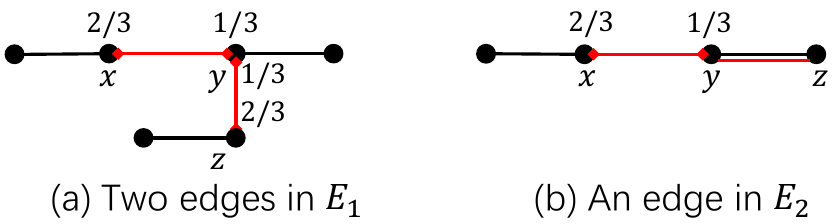}
    \caption{An illustration of the charging rule, where the edges in $M^*_{n/2}$ are represented by black edges and the edges in $E(P^*)$ are represented by red edges: In (a), we have a 3-path $xyz$ with $xy,yz\in E_1$, $e_x$, $e_y$ and $e_z$ are all charged by $\frac{2}{3}$ point by $xy$ and $yz$ in total; In (b), we have a 3-path $xyz$ with $xy\in E_2$, $e_x$ is charged by $\frac{2}{3}$ point by $xy$, and $e_y$ is charged by $\frac{1}{3}$ point by $xy$.}
    \label{fig:fig02}
\end{figure}

Assume that for any edge $e\in M^*_{n/2}$ it is charged by $\eta_e$ points in total.
Then, by Lemma~\ref{lb4}, we get $c(E_1)+c(E_2)\geq w(E_1)+w(E_2)-\sum_{e\in E(M^*_{n/2})}\eta_e\cdot w(e)$.
Therefore, we only need to calculate $\eta_e$ for each $e\in M^*_{n/2}$. We consider the following two cases.

\textbf{Case~1: $e=xy\in E(M^*_{n/2})\setminus E(P^*)$.}
Note that $x$ is either a 1-degree vertex or a 2-degree vertex. If it is a 1-degree vertex, it is charged by $\frac{2}{3}$ point by the edge of $E_1\cup E_2$ incident to $x$; otherwise, it is charged by $\frac{2}{3}$ point by two edges of $E_1\cup E_2$ incident to $x$. Hence, $xy$ is charged by $\frac{2}{3}$ point in total by the edges of $E_1\cup E_2$ incident to $x$. Similarly, $xy$ is charged by $\frac{2}{3}$ point in total by the edges of $E_1\cup E_2$ incident to $y$. So, $xy$ is charged by $\frac{4}{3}$ points in total.

\textbf{Case~2: $e\in E(M^*_{n/2})\cap E(P^*)$.}
It is easy to see that $e$ is charged by $\frac{1}{3}$ point by only one edge in $E_2$ (see example (b) in Fig.\ref{fig:fig02}).

Therefore, we get that 
\begin{align*}
\sum_{e\in E(M^*_{n/2})}\eta_e\cdot w(e) &= \sum_{e\in E(M^*_{n/2})\setminus E(P^*)}\frac{4}{3}w(e)+\sum_{e\in E(M^*_{n/2})\cap E(P^*)}\frac{1}{3}w(e)\\
&=\sum_{e\in E(M^*_{n/2})}\frac{4}{3}w(e)-\sum_{e\in E(M^*_{n/2})\cap E(P^*)}w(e)\\
&=\frac{4}{3}w(M^*_{n/2})-w(E(P^*_2)\cap E(M^*_{n/2})).
\end{align*}
Hence, we get $c(E_1)+c(E_2)\geq w(E_1) + w(E_2) + w(E(P^*_2)\cap E(M^*_{n/2})) - \frac{4}{3}w(M^*_{n/2})$.
\end{proof}
 
We are ready to obtain a bound related to $c(M^{**}_{n/6})$.
\begin{lemma}\label{lb6}
It holds that $c(M^{**}_{n/6})\geq \frac{1}{4}c(E_1)+\frac{1}{2}c(E_2)= \frac{1}{4}c(E(P^*)\setminus E(M^*_{n/2}))+\frac{1}{4}c(E_2)$.
\end{lemma}
\begin{proof}
We will construct a matching $M'_{n/6}$ of size $n/6$ in $G/M^*_{n/2}$ using the edges in $E_1\cup E_2$ such that $c(M'_{n/6})\geq \frac{1}{4}c(E_1)+\frac{1}{2}c(E_2)$.

Recall that $P^*=P^*_1\cup P^*_2$ and $E(P^*_1)\cap E(M^*_{n/2})=\emptyset$.
We first select $n/3$ edges in $E_1\cup E_2$ as follows. 
For each 3-path $xyz\in P^*_1$ with $c(xy)\geq c(yz)$, we select $xy$; for each 3-path $xyz\in P^*_2$ with $xy\notin E(M^*_{n/2})$, we select $xy$. We denote the set of these edges by $E'$. It holds that $c(E')\geq \frac{1}{2}c(E_1)+c(E_2)$ because $E_1=E(P^*_1)$ and $E_2= E(P^*_2)\setminus E(M^*_{n/2})$. Since $E'$ contains only one edge from each 3-path in $P^*$, these edges are vertex-disjoint. Moreover, since these edges are distinct from the edges in $E(M^*_{n/2})$, we know that $E'\cup E(M^*_{n/2})$ forms a set of vertex-disjoint paths and cycles in $G$. Denote the set of the paths (resp., cycles) by $P_{E'}$ (resp., $C_{E'}$).
As proved in~\cite{DBLP:journals/dam/Bar-NoyPRV18}, the cycles in $C_{E'}$ can be eliminated by updating the edges in $E_1\cap E'$ in a way satisfying that $E'\cup E(M^*_{n/2})$ forms a set of paths only, while maintaining $E'\cap E(M^*_{n/2})=\emptyset$, $\size{E'}=n/3$, and $c(E')\geq \frac{1}{2}c(E_1)+c(E_2)$.
We give a self-contained proof.

Let $e_v$ be the edge in $M^*_{n/2}$ containing $v$ as an endpoint.
We color the edges in $E'$ by red and the edges in $M^*_{n/2}$ by blue.
Assume that there is a cycle $C\in C_{E'}$. 
On the cycle $C$, the number of red edges, which is equal to the number of blue edges, is at least 2. This is because $E' \cap E(M^*_{n/2}) = \emptyset$, the edges in $E'$ are vertex-disjoint, and the edges in $E(M^*_{n/2})$ are also vertex-disjoint.
For each red edge $xy\in E(C)$, we assume that there is a 3-path $xyz\in P^*$. It is easy to see that $xy,yz\in E_1$; otherwise, we have $xy\in E_2$, $e_y=e_z=yz$, and the blue edge $yz$ will be incident to only one red edge in $E'\cup E(M^*_{n/2})$, a contradiction. Hence, we have $E'\cap E(C)=E_1\cap E(C)$. Moreover, we get that the blue edge $e_z$ is an end edge (i.e., the first or the last edge) of a path $P\in P_{E'}$. See Fig.\ref{fig:eliminate} for an illustration. 
\begin{figure}[t]
    \centering
    \includegraphics[scale=0.58]{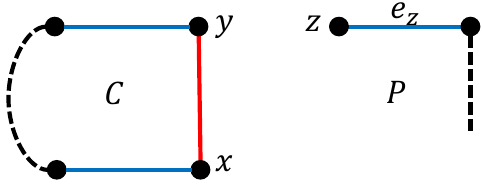}
    \caption{An illustration of a cycle $C\in C_{E'}$ and a 3-path $xyz\in P^*$, where $xy\in E(C)$, $e_z$ is an end edge of a path $P\in P_{E'}$, the edges in $M^{*}_{n/2}$ are colored blue, the edges in $E'$ are colored red, and the omitted edges in $E(C)\cup E(P)$ are represented by dotted black edges.}
    \label{fig:eliminate}
\end{figure}
Note that if we color $yz$ with red instead of $xy$, i.e., we select $yz$ instead of $xy$ when selecting $n/3$ edges in $E_1\cup E_2$, the cycle $C$ will be eliminated.

For each cycle $C\in C_{E'}$, we assume that $x_Cy_Cz_C\in P^*$ and $x_Cy_C$ is the red edge attaining 
\[
\arg\min_{xy\in E(C)\cap E',xyz\in P^*}\{c(xy)-c(yz)\}.
\]
Let $E''\coloneqq E'\setminus\{x_Cy_C\mid C\in C_{E'}\}\cup\{y_Cz_C\mid C\in C_{E'}\}$. Note that $E''$ is obtained by updating the edges in $E_1\cap E'$ and $E''$ still contains $n/3$ edges. It is easy to see that $E''\cup E(M^*_{n/2})$ forms a set of paths only. Next, we show that $c(E'')\geq \frac{1}{2}c(E_1)+c(E_2)$.

Let $E'_1\coloneqq\bigcup_{C\in C_{E'}}E'\cap E(C)$ and $E''_1\coloneqq E'\cap E_1\setminus E'_1$. Then, we have $E'\cap E_1=E'_1\cup E''_1$. Recall that $E'=E'\cap E_1\cup E'\cap E_2$ and  $E'\cap E_2=E_2$. We have $E'=E'_1\cup E''_1\cup E_2$. Note that $E(C)\cap E'=\{xy\mid xy\in E'\cap E(C),xyz\in P^*\}$.
Then, we have
\begin{align*}
c(E'')&=c(E')+\sum_{C\in C_{E'}}(c(y_Cz_C)-c(x_Cy_C))\\
&=c(E'_1)+c(E''_1)+c(E_2)+\sum_{C\in C_{E'}}(c(y_Cz_C)-c(x_Cy_C))\\
&=\sum_{C\in C_{E'}}\lrA{\sum_{\substack{xy\in E'\cap E(C),\\xyz\in P^*}}c(xy)+(c(y_Cz_C)-c(x_Cy_C))}+c(E''_1)+c(E_2)\\
&\geq\sum_{C\in C_{E'}}\lrA{\sum_{\substack{xy\in E'\cap E(C),\\xyz\in P^*}}c(xy)+(c(y_Cz_C)-c(x_Cy_C))}+\sum_{\substack{xy\in E''_1,\\xyz\in P^*}}\frac{1}{2}c(xyz)+c(E_2),
\end{align*}
where the inequality follows from $c(xy)\geq \frac{1}{2}c(xyz)$ since by the selection of $xy\in E'_1\cup E''_1$ with $xyz\in P^*$ we have $c(xy)\geq c(yz)$.

For each cycle $C\in C_{E'}$, there are at least two red edges in $E(C)$, and hence there exists another red edge $x'y'\in E(C)\cap E'$ and a 3-path $x'y'z'\in P^*$ such that $x'y'\neq x_Cy_C$ and $c(x_Cy_C)-c(y_Cz_C)\leq c(x'y')-c(y'z')$. Then, we have 
\begin{align*}
c(x'y') + c(y_Cz_C)&\geq \frac{1}{2}(c(x_Cy_C)+c(y_Cz_C)+c(x'y')+c(y'z'))\\
&=\frac{1}{2}(c(x_Cy_Cz_C)+c(x'y'z')).
\end{align*}
Hence, we have 
\begin{align*}
&\sum_{\substack{xy\in E'\cap E(C),\\xyz\in P^*}}c(xy)+(c(y_Cz_C)-c(x_Cy_C))\\
&=\sum_{\substack{xy\in E'\cap E(C)\setminus\{x_Cy_C,x'y'\},\\xyz\in P^*} }c(xy)+(c(x_Cy_C)+c(x'y')+c(y_Cz_C)-c(x_Cy_C))\\
&\geq \sum_{\substack{xy\in E'\cap E(C)\setminus\{x_Cy_C,x'y'\},\\xyz\in P^*}}\frac{1}{2}c(xyz)+\frac{1}{2}(c(x_Cy_Cz_C)+c(x'y'z'))=\sum_{\substack{xy\in E'\cap E(C),\\xyz\in P^*}}\frac{1}{2}c(xyz),
\end{align*}
where the inequality follows from $c(xy)\geq \frac{1}{2}c(xyz)$ since $E'\cap E(C)\subseteq E'_1$ and by the selection of $xy\in E'_1\cup E''_1$ with $xyz\in P^*$ we have $c(xy)\geq c(yz)$.

Then, we get 
\begin{align*}
c(E'')&\geq \sum_{C\in C_{E'}}\lrA{\sum_{\substack{xy\in E'\cap E(C),\\xyz\in P^*}}\frac{1}{2}c(xyz)}+\sum_{\substack{xy\in E''_1,\\xyz\in P^*}}\frac{1}{2}c(xyz)+c(E_2)\\
&=\sum_{\substack{xy\in E'_1,\\xyz\in P^*}}\frac{1}{2}c(xyz)+\sum_{\substack{xy\in E''_1,\\xyz\in P^*}}\frac{1}{2}c(xyz)+c(E_2)\\
&=\sum_{\substack{xy\in E'\cap E_1,\\xyz\in P^*}}\frac{1}{2}c(xyz)+c(E_2)=\frac{1}{2}c(E_1)+c(E_2),
\end{align*}
where the last equality follows from the fact that $E'\cap E_1$ is obtained by selecting an edge from each 3-path in $P^*_1$ and $E(P^*_1)=E_1$.

Therefore, given the optimal solution, we can construct a set of edges $E''$ in a way satisfying that $E''\cap E(M^*_{n/2})=\emptyset$ and $E''\cup E(M^*_{n/2})$ is a set of paths with $c(E'')\geq \frac{1}{2}c(E_1)+c(E_2)$ and $\size{E''}=n/3$. Hence, $E''$ is a set of paths in $G/M^*_{n/2}$, and it can be decomposed into two matchings of size $n/6$, as $n/3$ is even. Therefore, there exists a matching $M'$ of size $n/6$ in $G/M^*_{n/2}$ with a cost of at least $\frac{1}{2}c(E'')\geq \frac{1}{4}c(E_1)+\frac{1}{2}c(E_2)$. We get $c(M^{**}_{n/6})\geq c(M')\geq \frac{1}{4}c(E_1)+\frac{1}{2}c(E_2)$ by the optimality of $M^{**}_{n/6}$.
By the proof of Lemma~\ref{lb5}, we have $c(E_1)+c(E_2)=c(E(P^*)\setminus E(M^*_{n/2}))$, and hence $c(M^{**}_{n/6})\geq \frac{1}{4}c(E_1)+\frac{1}{2}c(E_2)= \frac{1}{4}c(E(P^*)\setminus E(M^*_{n/2}))+\frac{1}{4}c(E_2)$.
\end{proof}

Then, we are ready to analyze the approximation quality of Alg.1.
\begin{lemma}\label{lb8}
It holds that $w(P_1)\geq\frac{2}{3}w(M^*_{n/2})+\frac{1}{2}\cdot\OPT-\frac{1}{2}\sum_{xyz\in P^*}\max\{w(xy),w(yz)\}\geq \frac{7}{12}\cdot\OPT$.
\end{lemma}
\begin{proof}
Firstly, we prove that
\begin{equation}\label{lb-e2}
c(E_2)\geq \sum_{xyz\in P^*}(\min\{w(xy),w(yz)\}-\max\{w(xy),w(yz)\}).    
\end{equation}

For each 3-path $xyz\in P^*_2$, we assume that $xy\notin M^*_{n/2}$.
There are two edges $e_x$, $e_y\in M^*_{n/2}$ corresponding to the endpoints of $xy$ in $G/M^*_{n/2}$. Note that $e_y=yz$.
By Lemma~\ref{lb4}, we have 
\begin{align*}
c(E_2)&=\sum_{xyz\in P^*_2}c(xy)\geq \sum_{xyz\in P^*_2}(w(xy)-w(yz))\\
&\geq \sum_{xyz\in P^*_2}(\min\{w(xy),w(yz)\}-\max\{w(xy),w(yz)\})\\
&\geq \sum_{xyz\in P^*}(\min\{w(xy),w(yz)\}-\max\{w(xy),w(yz)\}), 
\end{align*}
where the first inequality follows from $c(xy)=w(xy)-\min\{w(e_x),w(e_y)\}\geq w(xy)-w(e_y)=w(xy)-w(yz)$, and the last from the fact that for any 3-path $xyz$ we have $\min\{w(xy),w(yz)\}-\max\{w(xy),w(yz)\}\leq 0$.

Thus, (\ref{lb-e2}) holds.

Then, we have that 
\begin{align*}
&w(P_1)\\
&\geq w(M^*_{n/2})+c(M^{**}_{n/6})\\
&\geq w(M^*_{n/2})+\frac{1}{4}c(E(P^*)\setminus E(M^*_{n/2}))+\frac{1}{4}c(E_2)\\
&\geq w(M^*_{n/2})+\frac{1}{4}\lrA{w(P^*)-\frac{4}{3}w(M^*_{n/2})}+\frac{1}{4}\sum_{xyz\in P^*}(\min\{w(xy),w(yz)\}-\max\{w(xy),w(yz)\})\\
&=\frac{2}{3}w(M^*_{n/2})+\frac{1}{4}\cdot\OPT+\frac{1}{4}\sum_{xyz\in P^*}(\min\{w(xy),w(yz)\}-\max\{w(xy),w(yz)\})\\
&=\frac{2}{3}w(M^*_{n/2})+\frac{1}{4}\cdot\OPT+\frac{1}{4}\sum_{xyz\in P^*}(\min\{w(xy),w(yz)\}+\max\{w(xy),w(yz)\}-2\max\{w(xy),w(yz)\})\\
&=\frac{2}{3}w(M^*_{n/2})+\frac{1}{4}\cdot\OPT+\frac{1}{4}\cdot\OPT-\frac{1}{2}\sum_{xyz\in P^*}\max\{w(xy),w(yz)\}\\
&\geq \frac{1}{6}w(M^*_{n/2})+\frac{1}{2}\sum_{xyz\in P^*}\max\{w(xy),w(yz)\}+\frac{1}{2}\cdot\OPT-\frac{1}{2}\sum_{xyz\in P^*}\max\{w(xy),w(yz)\}\\
&\geq \frac{1}{6}\cdot\frac{1}{2}\cdot\OPT+\frac{1}{2}\cdot\OPT=\frac{7}{12}\cdot\OPT,
\end{align*}
where the first inequality follows from Lemma~\ref{lb3}, the second from Lemma~\ref{lb6}, the third from Lemma~\ref{lb5} and (\ref{lb-e2}), the fourth from $w(M^*_{n/2})\geq \sum_{xyz\in P^*}\max\{w(xy),w(yz)\}$ by Lemma~\ref{lb1}, the last from $w(M^*_{n/2})\geq \frac{1}{2}\cdot\OPT$ by Lemma~\ref{lb1}.
\end{proof}

\begin{comment}
\begin{theorem}\label{tm1}
For MW3PP, Alg.1 is a polynomial-time $\frac{7}{12}$-approximation algorithm.
\end{theorem}
\begin{proof}
It is easy to see that the running time of Alg.1 in each loop is dominated by finding $M^*_{n/2}$ and $M^{**}_{n/6}$, which is $O(n^3)$. For the case that $n$ is even, we directly use Alg.1.
For the case that $n$ is odd, we first enumerate a 3-path in the optimal solution by taking $O(n^2)$ loops, and in each loop we run Alg.1, which takes $O(n^5)$ time in total. By Lemma~\ref{lb8}, we know that Alg.1 is a $\frac{7}{12}$-approximation algorithm with a running time of $O(n^5)$.
\end{proof}
\end{comment}

Note that the middle result $w(P_1)\geq\frac{2}{3}w(M^*_{n/2})+\frac{1}{2}\cdot\OPT-\frac{1}{2}\sum_{xyz\in P^*}\max\{w(xy),w(yz)\}$ in Lemma~\ref{lb8} will be used in our final analysis.
In the proof of Lemma~\ref{lb8}, we use the inequality $w(M^*_{n/2})\geq \frac{1}{2}\cdot\OPT$. However, we also have $w(M^*_{n/2})\geq w(M^*_{n/3})\geq\frac{1}{2}\cdot\OPT$ by Lemma~\ref{lb1}, which implies that in the worst case we may have $w(M^*_{n/2})=w(M^*_{n/3})$. Thus, it may be more effective to use $M^*_{n/3}$. 
%The cost performance of $M^*_{n/3}$ may be better.
Next, we design our second algorithm based on $M^*_{n/3}$.

\section{The Second Algorithm}\label{Second}
Our second algorithm is based on $M^*_{n/3}$, and we denote it by Alg.2. It contains four steps. 

\medskip
\noindent\textbf{Step~1.} Find $M^*_{n/3}$, a maximum weight matching of size $n/3$ in the graph $G$.
%, using $O(n^3)$ time~\cite{gabow1974implementation,lawler1976combinatorial}.

\noindent\textbf{Step~2.} Construct a multi-graph $G/M^*_{n/3}$ such that there are $n/3$ super-vertices one-to-one corresponding to the $n/3$ edges in $M^*_{n/3}$ and $n/3$ vertices in $V\setminus V(M^*_{n/3})$. We define the cost function $c$ on the edges in $G/M^*_{n/3}$ as follows:
\begin{itemize}
    \item For two super-vertices corresponding to edges $ux,yz\in M^*_{n/3}$, there are four edges $uy,uz,xy,xz$ between them, and each edge $e\in\{uy,uz,xy,xz\}$ is defined to have a cost of $c(e)\coloneqq w(e)-\min\{w(ux),w(yz)\}$.
    \item For one super-vertex corresponding to one edge $xy\in M^*_{n/3}$ and one vertex $z\in V\setminus V(M^*_{n/3})$, there are two edges $xz,yz$ between them, and each edge $e\in\{xz,yz\}$ is defined to have a cost of $c(e)\coloneqq w(e)$.
    \item For two vertices $x,y\in V\setminus V(M^*_{n/3})$, there is one edge $e=xy$ between them, and $e$ is defined to have a cost of $c(e)\coloneqq -\infty$. % w(e)
\end{itemize}

\noindent\textbf{Step~3.} Find a maximum cost matching $M^{**}$ of size at most $n/3$ in the graph $G/M^*_{n/3}$. 

\noindent\textbf{Step~4.} Obtain a 3-path packing $P_2$ of $G$ as follows:
\begin{itemize}
    \item For each edge $xy\in M^{**}$, 
    if there are two edges $e_x,e_y\in M^*_{n/3}$ corresponding to the endpoints of $xy$ in $G/M^*_{n/3}$, where we assume that $w(e_x)\geq w(e_y)$, obtain a 3-path using edges $xy$ and $e_x$, and call the vertex of $e_y$ that is not incident to $xy$ a \emph{residual vertex}.
    \item For each edge $xy\in M^{**}$, 
    if there is an edge $e_x\in M^*_{n/3}$ corresponding to the endpoint $x$ in $G/M^*_{n/3}$ and a vertex $y\in V\setminus V(M^*_{n/3})$, obtain a 3-path using edges $xy$ and $e_x$.
    \item For each vertex in $V\setminus V(M^*_{n/3})$, call it a \emph{residual vertex} if it is not matched by $M^{**}$. Arbitrarily create vertex-disjoint 2-paths using all residual vertices.
\end{itemize}
\medskip

In Step 3 of Alg.2, we do not find a maximum cost matching $M^{**}_{n/3}$ of size $n/3$. The reason is that we will construct a matching $M'$ in $G/M^*_{n/3}$ whose size may not be $n/3$, and then obtain a bound $c(M^{**})\geq c(M')$ by the optimality of $M^{**}$, while we may have $c(M^{**})>c(M^{**}_{n/3})$ since the cost function may be negative, and then we cannot get $c(M^{**}_{n/3})\geq c(M')$. 
Moreover, it holds that $c(M^{**})\geq 0$ since an empty matching has a cost of 0 and $M^{**}$ is a maximum cost matching. 

An illustration of Alg.2 can be seen in Fig.\ref{fig:fig03}.
\begin{figure}[t]
    \centering
    \includegraphics[scale=0.58]{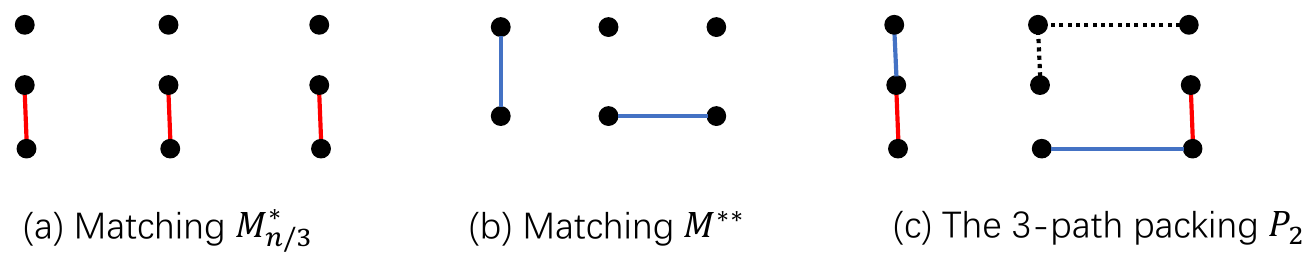}
    \caption{An illustration of Alg.2: In (a), $M^*_{n/3}$ contains three red edges; In (b), $M^{**}$ contains two blue edges; In (c), $G$ contains three residual vertices, and Alg.2 outputs three 3-paths.}
    \label{fig:fig03}
\end{figure}

\subsection{The analysis}
First, we consider the quality of the 3-path packing $P_2$ returned by Alg.2. For any $v\in V(M^*_{n/3})$, in this section, we let $e_v$ be the edge in $M^*_{n/3}$ containing $v$ as an endpoint. 
\begin{lemma}\label{lb10}
It holds that $w(P_2)\geq w(M^*_{n/3})+c(M^{**})$.
\end{lemma}
\begin{proof}
By Step 4 of Alg.2, there are three kinds of 3-paths in $P_2$, and we denote the set of the first kind of 3-paths in $P_2$ by $P'$ and the second kind of 3-paths in $P_2$ by $P''$. For each 3-path $xyz\in P'\cup P''$, we assume $xy\in M^*_{n/3}$ and $yz\in M^{**}$.
By definition of the cost function, for each 3-path $xyz\in P'$ we have $e_x=xy$, $w(e_x)\geq w(e_z)$, and $w(xy)+w(yz)=w(e_x)+w(e_z)-w(e_z)+w(yz)=w(e_x)+w(e_z)+c(yz)$ since $c(yz)=w(yz)-w(e_z)$; for each 3-path $xyz\in P''$ we have $e_x=xy$, and $w(xy)+w(yz)=w(e_x)+c(yz)$ since $w(yz)=c(yz)$. Therefore, we have 
\begin{align*}
&\sum_{xyz\in P'} (w(xy)+w(yz))+\sum_{xyz\in P''} (w(xy)+w(yz))\\
&=\sum_{xyz\in P'} (w(e_x)+w(e_z)+c(yz))+\sum_{xyz\in P''} (w(e_x)+c(yz))\\
&=w(M^*_{n/3})+c(M^{**}),
\end{align*}
where the second equality follows from the fact that by the assumption in Step 3 of Alg.2 for every edge in $M^*_{n/3}$ its corresponding vertex in $G/M^*_{n/3}$ is matched by $M^{**}$. Since the weight of the third kind of 3-paths is non-negative, we get $w(P_2)\geq w(P')+w(P'')=w(M^*_{n/3})+c(M^{**})$.
\end{proof}

Alg.2 can be seen as a simple variant of Alg.1. However, the main difficulty is to analyze its approximation ratio. We will generalize the use of the charging method in the proof of Lemma~\ref{lb5} to bound $c(M^{**})$.
Before that, we split $P^*$ into eight pairwise disjoint sets $P^*_1$, $P^*_2$, $P^*_3$, $P^*_4$, $P^*_5$, $P^*_6$, $P^*_7$, and $P^*_8$ based on $M^*_{n/3}$: consider a 3-path $xyz\in P^*$ and we have 
\begin{align*}
\left\{
\begin{array}{*{20}l}
xyz\in P^*_1, & \size{\{x,y,z\}\cap V(M^*_{n/3})}=0,\\
xyz\in P^*_2, & \size{\{x,y,z\}\cap V(M^*_{n/3})}=1~\mbox{and}~y\notin V(M^*_{n/3}),\\
xyz\in P^*_3, & \size{\{x,y,z\}\cap V(M^*_{n/3})}=1~\mbox{and}~y\in V(M^*_{n/3}),\\
xyz\in P^*_4, & \size{\{x,y,z\}\cap V(M^*_{n/3})}=2~\mbox{and}~y\notin V(M^*_{n/3}),\\
xyz\in P^*_5, & \size{\{x,y,z\}\cap V(M^*_{n/3})}=2~\mbox{and}~y\in V(M^*_{n/3})~\mbox{and}~\size{\{xy,yz\}\cap E(M^*_{n/3})}=1,\\
xyz\in P^*_6, & \size{\{x,y,z\}\cap V(M^*_{n/3})}=2~\mbox{and}~y\in V(M^*_{n/3})~\mbox{and}~\size{\{xy,yz\}\cap E(M^*_{n/3})}=0,\\
xyz\in P^*_7, & \size{\{x,y,z\}\cap V(M^*_{n/3})}=3~\mbox{and}~\size{\{xy,yz\}\cap E(M^*_{n/3})}=1,\\
xyz\in P^*_8, & \size{\{x,y,z\}\cap V(M^*_{n/3})}=3~\mbox{and}~\size{\{xy,yz\}\cap E(M^*_{n/3})}=0.
\end{array}
\right.
\end{align*}

We assume that the optimal solution $P^*$ satisfies the following property.
\begin{assumption}\label{ASS}
For each 3-path $xyz$ in $P^*$ with $y\notin V(M^*_{n/3})$ and $x,z\in V(M^*_{n/3})$, the vertices $x$ and $z$ are contained in two distinct edges of $M^*_{n/3}$.
\end{assumption}

If Assumption~\ref{ASS} does not hold, we have $xz\in M^*_{n/3}$, and then $(P^*\setminus\{xyz\})\cup\{xzy\}$ is still an optimal solution by the optimality of $M^*_{n/3}$. Hence, there must exist an optimal solution satisfying the property. Note that Assumption~\ref{ASS} was made to ensure that there is only one kind of 3-path in $P^*_4$, as illustrated in Fig.\ref{fig:fig04}, which will simplify our analysis.

Next, we define the following notation.
\begin{itemize}
    \item For $P^*_i$, where $i\in\{1,3,4,8\}$, define $X_i\coloneqq \{xy\mid xyz\in P^*_i, w(xy)\geq w(yz)\}$.
    
    \item For $P^*_2$, define $X_2\coloneqq \{xy\mid xyz\in P^*_2, x\in V(M^*_{n/3})\}$.
    
    \item For $P^*_5$, define $X_5\coloneqq \{xy\mid xyz\in P^*_5, xy\notin E(M^*_{n/3})\}$.

    \item For $P^*_6$, define $X_6\coloneqq \{xy\mid xyz\in P^*_6, x\notin V(M^*_{n/3})\}$.

    \item For $P^*_7$, define $X_7\coloneqq \{xy\mid xyz\in P^*_7, xy\notin E(M^*_{n/3})\}$.

    \item For $P^*_i$, where $i\in\{1,2,...,8\}$, define $Y_i\coloneqq E(P^*_i)\setminus X_i$.
\end{itemize}
Note that we have $E(P^*_i)=X_i\cup Y_i$ and $X_i\cap Y_i=\emptyset$ for each $i\in\{1,...,8\}$.

For an edge $xy\in E(P^*)$, it is called a \emph{middle} edge if it satisfies that $\{x,y\}\cap V(M^*_{n/3})=1$; otherwise, it is called a \emph{non-middle} edge. Note that all middle edges are in $E(P^*_2\cup P^*_3\cup P^*_4\cup P^*_5\cup P^*_6)$. 
For edges in $M^*_{n/3}$, we split them into five pairwise disjoint sets: $M_1$, $M_2$, $M_3$, $M_4$, and $M_5$.
\begin{itemize}
    \item Define $M_1$ to be the set of edges in $E(M^*_{n/3})$ connecting at least one middle edge in $E(P^*_6)$.
    \item Define $M_2$ to be the set of edges in $E(M^*_{n/3})$ connecting at least one middle edge, where the middle edge is only contained in $E(P^*_2\cup P^*_3\cup P^*_4)$.
    \item Define $M_3\coloneqq E(M^*_{n/3})\cap E(P^*_7)$.
    \item Define $M_4\coloneqq E(M^*_{n/3})\cap E(P^*_5)$.
    \item Define $M_5\coloneqq E(M^*_{n/3})\setminus (M_1\cup M_2\cup M_3\cup M_4)$.
\end{itemize}

See Fig.\ref{fig:fig04} for an illustration of the 8 kinds of 3-paths in $P^*$, and the sets $X_i$, $Y_i$, and $M_i$.

\begin{comment}
\begin{align*}
\left\{
\begin{array}{*{20}l}
xyz\in P^*_1, & \{x,y,z\}\cap V(M^*_{n/3})=0,\\
xyz\in P^*_2, & \{x,y,z\}\cap V(M^*_{n/3})=1,\\
xyz\in P^*_3, & \{x,y,z\}\cap V(M^*_{n/3})=0,\\
xyz\in P^*_4, & \{x,y,z\}\cap V(M^*_{n/3})=0,\\
\end{array}
\right.&
\quad\quad\quad\quad
\left\{
\begin{array}{*{20}l}
xyz\in P^*_5, & \{x,y,z\}\cap V(M^*_{n/3})=0,\\
xyz\in P^*_6, & \{x,y,z\}\cap V(M^*_{n/3})=0,\\
xyz\in P^*_7, & \{x,y,z\}\cap V(M^*_{n/3})=0,\\
xyz\in P^*_8, & \{x,y,z\}\cap V(M^*_{n/3})=0,\\
\end{array}
\right.
\end{align*}

\begin{enumerate}
    \item[$xyz\in P^*_1$:] $\{x,y,z\}\cap V(M^*_{n/3})$;
    \item[$xyz\in P^*_2$:] $\{x,y,z\}\notin V(M^*_{n/3})$;
    \item[$xyz\in P^*_3$:] the set of all type-2 partial-external triangles in $\B^*$;
    \item[$xyz\in P^*_4$:] the set of all type-1 external triangles in $\B^*$;
    \item[$xyz\in P^*_5$:] the set of all type-2 external triangles in $\B^*$.
\end{enumerate}
\end{comment}

\begin{figure}[t]
    \centering
    \includegraphics[scale=0.58]{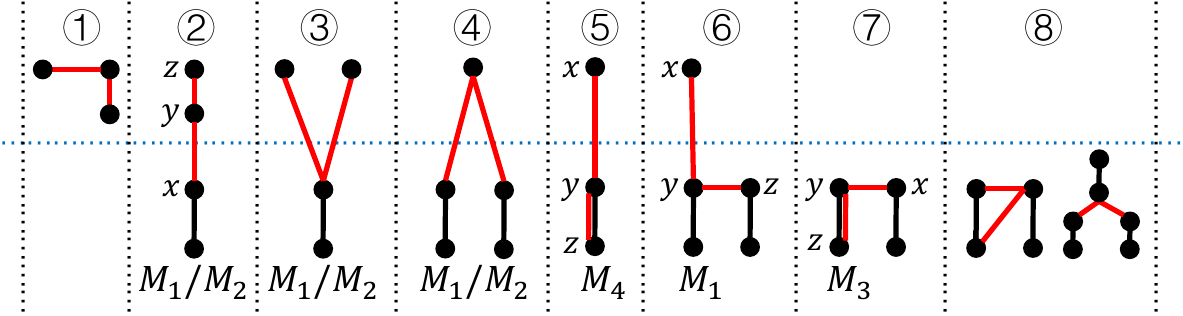}
    \caption{An illustration of the eight kinds of 3-paths in the optimal solution $P^*$, where the black edges represent the edges in $M^*_{n/3}$, the red edges represent the edges of 3-paths in $P^*$, and the 3-paths on the $i$-th column represent the 3-paths in $P^*_i$.}
    \label{fig:fig04}
\end{figure}

\begin{lemma}\label{lb11}
It holds that $w(X_i)\geq w(Y_i)$ for each $i\in\{1,3,4,8\}$, $w(M_3)=w(Y_7)$, $w(M_4)=w(Y_5)$, and $w(Y_5)\geq w(X_5)$.    
\end{lemma}
\begin{proof}
For each $i\in\{1,3,4,8\}$, we have $X_i=\{xy\mid xyz\in P^*_i, w(xy)\geq w(yz)\}$ and $Y_i=\{yz\mid xyz\in P^*_i, w(xy)\geq w(yz)\}$ by definition, and hence we have $w(X_i)\geq w(Y_i)$. We also have $M_3=Y_7$ and $M_4=Y_5$ by definition, and then we have $w(M_3)=w(Y_7)$ and $w(M_4)=w(Y_5)$. Since $(E(M^*_{n/3})\setminus Y_5)\cup X_5$ is a matching of size $n/3$ by definition, we have $w(Y_5)\geq w(X_5)$ since $w(M^*_{n/3})\geq w(M^*_{n/3})-w(Y_5)+w(X_5)$ by the optimality of $M^*_{n/3}$.
\end{proof}

\begin{lemma}\label{lb13}
It holds that $w(M^*_{n/2})\geq w(M^*_{n/3})+w(X_1)+w(Y_2)$.
\end{lemma}
\begin{proof}
By definition, we know that three sets $V(X_1)$, $V(Y_2)$, and $V(M^*_{n/3})$ are pairwise disjoint, and $E(M^*_{n/3})\cup X_1\cup Y_2$ is a matching of size at most $n/2$ in $G$. Since $M^*_{n/2}$ is the maximum weight matching in $G$, we have $w(M^*_{n/2})\geq w(M^*_{n/3})+w(X_1)+w(Y_2)$.
\end{proof}

\begin{lemma}\label{lb14}
It holds that $w(M^*_{n/2})\geq w(M^*_{n/3})-w(M_1)-w(M_2)+w(X_1)+\sum_{xyz\in P^*_2}\max\{w(xy),w(yz)\}+w(X_3)+w(X_4)+w(X_6)$.
\end{lemma}
\begin{proof}
Since the edges in $M_1\cup M_2$ are the edges in $M^*_{n/3}$ connecting all middle edges in $E(P^*_2\cup P^*_3\cup P^*_4\cup P^*_6)$, we know that two sets $V(M^*_{n/3}\setminus(M_1\cup M_2))$ and $V(P^*_1\cup P^*_2\cup P^*_3\cup P^*_4\cup P^*_6)$ are disjoint. Let $Z_2\coloneqq\{xy\mid xyz\in E(P^*_2),w(xy)\geq w(yz)\}$ and we have $w(Z_2)=\sum_{xyz\in P^*_2}\max\{w(xy),w(yz)\}$. Note that $(M^*_{n/3}\setminus(M_1\cup M_2))\cup (X_1\cup Z_2\cup X_3\cup X_4\cup X_6)$ is a matching of size at most $n/2$ in $G$. Since $M^*_{n/2}$ is the maximum weight matching in $G$, by definitions of $M_i$, $X_i$, and $Y_i$, it is easy to see that the lemma holds.
\end{proof}

\begin{lemma}\label{lb15}
It holds that $w(M^*_{n/2})\geq w(M^*_{n/3})-w(M_1)+w(X_1)+w(Y_2)+w(X_6)$.
\end{lemma}
\begin{proof}
Since the edges in $M_1$ are the edges in $M^*_{n/3}$ connecting all middle edges in $E(P^*_6)$, we know two sets $V(M^*_{n/3}\setminus M_1)$ and $V(P^*_1)\cup V(Y_2)\cup V(X_6)$ are disjoint. Note that $(M^*_{n/3}\setminus M_1)\cup X_1\cup Y_2\cup X_6$ is a matching of size at most $n/2$ in $G$. Since $M^*_{n/2}$ is the maximum weight matching in $G$, by definitions of $M_i$, $X_i$, and $Y_i$, it is easy to see that the lemma holds.
\end{proof}

Recall that $M_1$ is the set of edges in $E(M^*_{n/3})$ connecting at least one middle edge in $E(P^*_6)$. It connects zero middle edges in $E(P^*_5)$, but may connect middle edges in $E(P^*_2\cup P^*_3\cup P^*_4)$. We need to further split $M_1$ into three pairwise disjoint sets: $M'_1$, $M''_1$, and $M'''_1$.
\begin{itemize}
    \item Define $M'_1$ to be the set of edges in $E(M^*_{n/3})$ connecting one middle edge in $E(P^*_6)$ and zero middle edges in $E(P^*_2\cup P^*_3\cup P^*_4)$.
    \item Define $M''_1$ to be the set of edges in $E(M^*_{n/3})$ connecting two middle edges in $E(P^*_6)$ and zero middle edges in $E(P^*_2\cup P^*_3\cup P^*_4)$.
    \item Define $M'''_1$ to be the set of edges in $E(M^*_{n/3})$ connecting one middle edge in $E(P^*_6)$ and at least one middle edge in $E(P^*_2\cup P^*_3\cup P^*_4)$.
\end{itemize} 

\begin{lemma}\label{lb16}
It holds that $w(M^*_{n/2})\geq w(M^*_{n/3})-w(M'''_1)-w(M_2)+w(X_1)+\sum_{xyz\in P^*_2}\max\{w(xy),w(yz)\}+w(X_3)+w(X_4)$.
\end{lemma}
\begin{proof}
Since the edges in $M'''_1\cup M_2$ are the edges in $M^*_{n/3}$ connecting all middle edges in $E(P^*_2\cup P^*_3\cup P^*_4)$, we know that two sets $V(M^*_{n/3}\setminus(M'''_1\cup M_2))$ and $V(P^*_1\cup P^*_2\cup P^*_3\cup P^*_4)$ are disjoint. Let $Z_2\coloneqq\{xy\mid xyz\in E(P^*_2),w(xy)\geq w(yz)\}$ and we have $w(Z_2)=\sum_{xyz\in P^*_2}\max\{w(xy),w(yz)\}$. Note that $(M^*_{n/3}\setminus(M'''_1\cup M_2))\cup (X_1\cup Z_2\cup X_3\cup X_4)$ is a matching of size at most $n/2$ in $G$. Since $M^*_{n/2}$ is the maximum weight matching in $G$, by definitions of $M'''_1$, $M_2$, $X_i$, and $Y_i$, it is easy to see that the lemma holds.
\end{proof}

Lemmas~\ref{lb11}-\ref{lb16} state some inequalities satisfied by the weights of our newly defined edge sets and the maximum weight matchings $M^*_{n/3}$ and $M^*_{n/2}$. 

To analyze Alg.2, by Lemma~\ref{lb10}, we need to obtain a lower bound on $c(M^{**})$. 

We use the edges in $E(P^*_2\cup P^*_3\cup P^*_4\cup P^*_5\cup P^*_6\cup P^*_7\cup P^*_8)\setminus E(M^*_{n/3})$ to obtain three pairwise disjoint sets: $E_1$, $E_2$, and $E_3$ in the following way.
\begin{itemize}
    \item Define $E_1\coloneqq E(P^*_6)\cup E(P^*_8)$.
    \item Define $E_2\coloneqq X_2\cup X_3\cup X_4\cup X_7$.
    \item Define $E_3\coloneqq X_5$.
\end{itemize}
We will use the edges in $E_1\cup E_2\cup E_3$ to construct a matching $M'$ in $G/M^*_{n/3}$, and then derive a bound $c(M^{**})\geq c(M')$ by the optimality of $M^{**}$.

Next, we first analyze the cost of the edges in $E_1\cup E_2\cup E_3$.

\begin{lemma}\label{lb17}
It holds that $c(E_2)\geq w(X_2)+w(X_3)+w(X_4)+w(X_7)-w(Y_7)$, and $c(E_3)=w(X_5)$.
\end{lemma}
\begin{proof}
Note that $c(E_2)=c(X_2)+c(X_3)+c(X_4)+c(X_7)$ and $c(E_3)=c(X_5)$.
Since the cost of a middle edge equals to its weight, we have $c(X_i)=w(X_i)$ for each $i\in\{2,3,4\}$. 
Recall that $X_7=\{xy\mid xyz\in P^*_7, xy\notin E(M^*_{n/3})\}$. By Lemma~\ref{lb4}, we have 
\begin{align*}
c(X_7)&=\sum_{xyz\in P^*_7,xy\notin E(M^*_{n/3})}c(xy)\\
&\geq\sum_{xyz\in P^*_7,xy\notin E(M^*_{n/3})}(w(xy)-w(yz))\\
&=w(X_7)-w(Y_7).
\end{align*}
Hence, we have $c(E_2)\geq w(X_2)+w(X_3)+w(X_4)+w(X_7)-w(Y_7)$.
\end{proof}

\begin{lemma}\label{lb18}
We have $c(E_1)+c(E_2)+c(E_3)\geq w(P^*_6)+w(P^*_8)+w(X_2)+w(X_3)+w(X_4)+w(X_5)+w(X_7)-\frac{4}{3}w(M^*_{n/3})+\frac{1}{3}w(M'_1)+\frac{2}{3}w(M''_1)+w(M'''_1)+\frac{2}{3}w(M_2)+w(M_3)+\frac{4}{3}w(M_4)$.
\end{lemma}
\begin{proof}
To prove this lemma, we still use the charging method in Lemma~\ref{lb5}.
For any non-middle edge $xy$ in $E_1\cup E_2\cup E_3$, there are two edges $e_x,e_y\in M^*_{n/3}$ corresponding to the endpoints of $xy$ in $G/M^*_{n/3}$. By Lemma~\ref{lb4}, it holds that $c(xy)\geq w(xy)-(\theta\cdot w(e_x)+(1-\theta)\cdot w(e_y))$ for any $0\leq \theta\leq 1$, and we say that $xy$ charges $\theta$ point from $e_x$ and $1-\theta$ point from $e_y$. Note that if $xy$ is a middle edge we have $c(xy)=w(xy)$, i.e., it does not charge any point.
Next, we design a charging rule so that each non-middle edge in $E_1\cup E_2\cup E_3$ receives 1 point in total.

A vertex is called an \emph{$i$-degree vertex} if it is incident to $i$ edges in $E(P^*)$. For any edge $xy\in E_1\cup E_2\cup E_3$, there exists exactly one 2-degree vertex and one 1-degree vertex in $\{x,y\}$. Assume w.l.o.g. that $x$ is a 1-degree vertex and $y$ is a 2-degree vertex. If $xy$ is a non-middle edge, we let $xy$ charge $\frac{2}{3}$ point from $e_x$ and $\frac{1}{3}$ point from $e_y$. Note that $e_x\neq e_y$.

Assume that for any edge $e\in M^*_{n/3}$ it is charged by $\eta_e$ points in total by non-middle edges in $E_1\cup E_2\cup E_3$.
Then, by Lemma~\ref{lb4}, we get $c(E_1)+c(E_2)+c(E_3)\geq w(E_1)+w(E_2)+w(E_3)-\sum_{e\in E(M^*_{n/3})}\eta_e\cdot w(e)$.
Next, we calculate $\eta_e$ for each $e\in M^*_{n/3}$ by considering the following cases. We only consider non-middle edges in $E_1\cup E_2\cup E_3$.

\textbf{Case~1: $e\in M_1$.} Note that $M_1=M'_1\cup M''_1\cup M'''_1$. We further consider three cases.

\textbf{Case~1.1: $e=xy\in M'_1$.} By definition, $e$ connects one middle edge in $E(P^*_6)$ and no other middle edges. Assume that $x$ is incident to one middle edge in $E(P^*_6)\subseteq E_1$. We know that $x$ is a 2-degree vertex, and $x$ is also incident to one non-middle edge in $E(P^*_6)$ that charges $xy$ $\frac{1}{3}$ point. Note that $y$ must be incident to non-middle edges only in $E_1\cup E_2$.
If $y$ is incident to one non-middle edge, $y$ is a 1-degree vertex, and $xy$ will be charged by $\frac{2}{3}$ point by the incident edge; otherwise, it is incident to two non-middle edges, $y$ is a 2-degree vertex, and $xy$ will be charged by $\frac{2}{3}$ point in total by the two incident edges. Hence, $xy$ is charged by $1$ point in total, and we have $\eta_e=1$.

\textbf{Case~1.2: $e=xy\in M''_1$.} By definition, $e$ connects two middle edges in $E(P^*_6)$ and no other middle edges. Note that a vertex of $e$ can be incident to at most one middle edge in $E(P^*_6)$. Hence, both $x$ and $y$ are incident to one middle edge in $E(P^*_6)$, both are 2-degree vertices, and incident to one non-middle middle in $E(P^*_6)$. By the analysis of Case~1.1, we know that $\eta_e=\frac{1}{3}+\frac{1}{3}=\frac{2}{3}$.

\textbf{Case~1.3: $e=xy\in M'''_1$.} By definition, $e$ connects one middle edge in $E(P^*_6)$ and at least one middle edge in $E(P^*_2\cup P^*_3\cup P^*_4)$. If $x$ is incident to the middle edge in $E(P^*_6)$, it is also incident to one non-middle edge in $E(P^*_6)$ that charges $xy$ $\frac{1}{3}$ point. Hence, $y$ is incident to at least one middle edge in $E(P^*_2\cup P^*_3\cup P^*_4)$, and then it cannot be incident to any non-middle edge. Hence, $y$ is incident to only middle edges. Therefore, $xy$ is charged by $\frac{1}{3}$ point in total, and we have $\eta_e=\frac{1}{3}$.

\textbf{Case~2: $e=xy\in M_2$.} By definition, $xy$ connects middle edges only in $E(P^*_2\cup P^*_3\cup P^*_4)$. Assume that $x$ is incident to one middle edge in $E(P^*_2\cup P^*_3\cup P^*_4)$. As mentioned in the analysis of Case~1.3, $x$ cannot be incident to non-middle edges. Similarly, if $y$ is incident to one middle edge, and then it cannot be incident to non-middle edges. If $y$ is incident to non-middle edges, it may be incident to one or two non-middle edges, and by the analysis of Case~1.1 $xy$ will be charged by $\frac{2}{3}$ point in total by the edges incident to $y$. Hence, $xy$ is charged by at most $\frac{2}{3}$ point by the possible existing non-middle edges incident to $y$, and we have $\eta_e\leq\frac{2}{3}$.

\textbf{Case~3: $e=xy\in M_3$.} By definition, it is easy to get that $xy$ connects only one non-middle edge and zero middle edges (see Fig.\ref{fig:fig04}). Moreover, we assume that $x$ is incident to one non-middle edge, and we get that $x$ is a 2-degree vertex. Hence, $xy$ is charged by $\frac{1}{3}$ point by the non-middle edge, and we have $\eta_e=\frac{1}{3}$. 

\textbf{Case~4: $e=xy\in M_4$.} By definition, $xy$ connects only one middle edge and zero non-middle edges (see Fig.\ref{fig:fig04}). Hence, $xy$ is not charged, and we have $\eta_e=0$. 

\textbf{Case~5: $e=xy\in M_5$.} By definition, $xy$ connects only non-middle edges and zero middle edges. Note that $x$ (or $y$) may be incident to one or two non-middle edges. By the analysis of Case~1.1, $xy$ will be charged by $\frac{2}{3}$ point in total by the edges incident to $x$ (or $y$). Hence, $xy$ is charged by $\frac{4}{3}$ points by the non-middle edges incident to $x$ and $y$, and we have $\eta_e=\frac{4}{3}$.

Therefore, we get that 
\begin{align*}
&\sum_{e\in E(M^*_{n/3})}\eta_e\cdot w(e)\\
&\leq \sum_{e\in M'_1}w(e)+\sum_{e\in M''_1}\frac{2}{3}w(e)+\sum_{e\in M'''_1}\frac{1}{3}w(e)+\sum_{e\in M_2}\frac{2}{3}w(e)+\sum_{e\in M_3}\frac{1}{3}w(e)+\sum_{e\in M_5}\frac{4}{3}w(e)\\
&=\frac{4}{3}w(M^*_{n/3})-\frac{1}{3}w(M'_1)-\frac{2}{3}w(M''_1)-w(M'''_1)-\frac{2}{3}w(M_2)-w(M_3)-\frac{4}{3}w(M_4),
\end{align*}
where the last equality follows from the fact that $E(M^*_{n/3})=M_1\cup\cdots\cup M_5$.

Recall that $c(E_1)+c(E_2)+c(E_3)\geq w(E_1)+w(E_2)+w(E_3)-\sum_{e\in E(M^*_{n/3})}\eta_e\cdot w(e)$. By definitions of $E_1$, $E_2$, and $E_3$, we have $w(E_1)=w(P^*_6)+w(P^*_8)$, $w(E_2)=w(X_2)+w(X_3)+w(X_4)+w(X_7)$, and $w(E_3)=w(X_5)$. Hence, we get that $c(E_1)+c(E_2)+c(E_3)\geq w(P^*_6)+w(P^*_8)+w(X_2)+w(X_3)+w(X_4)+w(X_5)+w(X_7)-\frac{4}{3}w(M^*_{n/3})+\frac{1}{3}w(M'_1)+\frac{2}{3}w(M''_1)+w(M'''_1)+\frac{2}{3}w(M_2)+w(M_3)+\frac{4}{3}w(M_4)$.
\end{proof}

Note that the charging method involves only two cases in the proof Lemma~\ref{lb5} for Alg.1, which correspond to two of five cases (Cases~3 and~5) in the proof of Lemma~\ref{lb18} for Alg.2.

Based on the proof of Lemma~\ref{lb6}, we are ready to get a lower bound on $c(M^{**})$.
\begin{lemma}\label{lb19}
It holds that $c(M^{**})\geq \frac{1}{4}c(E_1)+\frac{1}{2}c(E_2)+c(E_3)$.
\end{lemma}
\begin{proof}
We will construct a matching $M'$ in $G/M^*_{n/3}$ using the edges in $E_1\cup E_2\cup E_3$ such that $c(M')\geq \frac{1}{4}c(E_1)+\frac{1}{2}c(E_2)+c(E_3)$.

We first select $\frac{1}{2}\size{E_1}+\size{E_2}+\size{E_3}$ edges in $E_1\cup E_2\cup E_3$ as follows. Select $\frac{1}{2}\size{E_1}$ edges from $E_1$ by selecting $xy$ for each 3-path $xyz\in P^*_6\cup P^*_8$ with $c(xy)\geq c(yz)$; Select all edges in $E_2$ and $E_3$. Denote the selected edges by $E'$, and clearly we have $c(E')\geq\frac{1}{2}c(E_1)+c(E_2)+c(E_3)$.
Note that we have $\size{E'}=\frac{1}{2}\size{E_1}+\size{E_2}+\size{E_3}=\frac{1}{2}\sum_{i=2}^{8}\size{E(P^*_i)}$, i.e., $E'$ contains one edge from each 3-path in $P^*_2\cup\cdots\cup P^*_8$ and zero edges from each 3-path in $P^*_1$. Hence, these edges are vertex-disjoint. Moreover, since these edges are distinct from the edges in $E(M^*_{n/3})$, we know that $E'\cup E(M^*_{n/3})$ forms a set of paths and cycles in $G$.
Recall that $E_2\cup E_3=X_2\cup X_3\cup X_4\cup X_5\cup X_7$.
Since each edge in $X_2\cup X_3\cup X_4\cup X_5$ (resp., $X_7$) must be the first/last edge (resp., the second edge to the first/last edge) of a path in $E'\cup E(M^*_{n/3})$, if there is a cycle $C$ formed by $E'\cup E(M^*_{n/3})$, we have $E(C)\cap E'=E(C)\cap E_1$. By the proof of Lemma~\ref{lb6}, we know that we can eliminate all cycles formed by $E'\cup E(M^*_{n/3})$ by updating $E_1\cap E'$, while maintaining $c(E')\geq\frac{1}{2}c(E_1)+c(E_2)+c(E_3)$. Note that we only update the edges in $E_1\cap E'$. Hence, we still have $E_2\cup E_3\subseteq E'$, and then we get $c(E'\cap E_1)\geq \frac{1}{2}c(E_1)$ since $c(E')=c(E'\cap E_1)+c(E_2)+c(E_3)$.

Therefore, given the optimal solution, we can construct a set of edges $E'$ satisfying that $E'\cap E(M^*_{n/3})=\emptyset$ and $E'\cup E(M^*_{n/3})$ forms a set of paths only with $c(E')\geq \frac{1}{2}c(E_1)+c(E_2)+c(E_3)$, and then $E'$ forms a set of paths in $G/M^*_{n/3}$. Moreover, every edge in $E_3$ is a single isolated path, and the edges in $(E'\cap E_1)\cup E_2$ form a set of paths. Hence, by decomposing the edges in $(E'\cap E_1)\cup E_2$ into two matchings and combining the matching with the edges in $E_3$, we get a matching $M'$ with a cost of at least $\frac{1}{2}c((E'\cap E_1)\cup E_2)+c(E_3)=\frac{1}{2}c(E'\cap E_1)+\frac{1}{2}c(E_2)+c(E_3)\geq \frac{1}{4}c(E_1)+\frac{1}{2}c(E_2)+c(E_3)$ since $c(E'\cap E_1)\geq \frac{1}{2}c(E_1)$.

By the optimality of $M^{**}$, we get $c(M^{**})\geq c(M')\geq\frac{1}{4}c(E_1)+\frac{1}{2}c(E_2)+c(E_3)$.
\end{proof}
%By the previous 
\begin{lemma}\label{lb20}
It holds that $w(P_2)\geq\frac{2}{3}w(M^*_{n/3})+\frac{1}{12}w(M'_1)+\frac{1}{6}w(M''_1)+\frac{1}{4}w(M'''_1)+\frac{1}{6}w(M_2)+\frac{1}{4}w(M_3)+\frac{1}{3}w(M_4)+\frac{1}{2}w(X_2)+\frac{1}{2}w(X_3)+\frac{1}{2}w(X_4)+w(X_5)+\frac{1}{2}w(X_7)-\frac{1}{4}w(Y_7)+\frac{1}{4}w(P^*_6)+\frac{1}{4}w(P^*_8)$.
\end{lemma}
\begin{proof}
We get that 
\begin{align*}
&w(P_2)\\
&\geq w(M^*_{n/3})+c(M^{**})\\
&\geq w(M^*_{n/3})+\frac{1}{4}c(E_1)+\frac{1}{2}c(E_2)+c(E_3)\\
&=w(M^*_{n/3})+\frac{1}{4}(c(E_1)+c(E_2)+c(E_3))+\frac{1}{4}c(E_2)+\frac{3}{4}c(E_3)\\
&\geq w(M^*_{n/3})+\frac{1}{4}\left(w(P^*_6)+w(P^*_8)+w(X_2)+w(X_3)+w(X_4)+w(X_5)+w(X_7)-\frac{4}{3}w(M^*_{n/3})\right.\\
&\quad\left.\!+\frac{1}{3}w(M'_1)+\frac{2}{3}w(M''_1)+w(M'''_1)+\frac{2}{3}w(M_2)+w(M_3)+\frac{4}{3}w(M_4)\right)\\
&\quad+\frac{1}{4}(w(X_2)+w(X_3)+w(X_4)+w(X_7)-w(Y_7))+\frac{3}{4}w(X_5)\\
&=\frac{2}{3}w(M^*_{n/3})+\frac{1}{12}w(M'_1)+\frac{1}{6}w(M''_1)+\frac{1}{4}w(M'''_1)+\frac{1}{6}w(M_2)+\frac{1}{4}w(M_3)+\frac{1}{3}w(M_4)\\
&\quad+\frac{1}{2}w(X_2)+\frac{1}{2}w(X_3)+\frac{1}{2}w(X_4)+w(X_5)+\frac{1}{2}w(X_7)-\frac{1}{4}w(Y_7)+\frac{1}{4}w(P^*_6)+\frac{1}{4}w(P^*_8),
\end{align*}
where the first inequality follows from Lemma~\ref{lb10}, the second from Lemma~\ref{lb19}, and the last from Lemmas~\ref{lb17} and \ref{lb18}.
\end{proof}

%\begin{theorem}\label{tm2}
%For MW3PP, Alg.2 is a polynomial-time algorithm.
%\end{theorem}

We do not theoretically prove an approximation ratio better than $\frac{7}{12}$ by using the results in our previous lemmas and making a trade-off between Alg.1 and Alg.2. However, if we make a trade-off between Alg.1 and Alg.2, we find that $\sum_{i=1}^{6}w(P^*_i)/\OPT<10^{-4}$ in the worst case.\footnote{Based on the linear program in the proof of Theorem~\ref{main}, numerical results show that, under the constraint that $\sum_{i=1}^{6}w(P^*_i)/\OPT\geq 10^{-4}$, we can obtain an approximation ratio of at least $0.583335>\frac{7}{12}$ by making a trade-off between Alg.1 and Alg.2.} This implies that almost all weighted parts of $P^*$ lie within the graph $G[V(M^*_{n/3})]$. Therefore, the maximum weight 2-star packing may have a weight of at least $\OPT$ in this graph. We will design our third algorithm to handle this case.

\section{The Third Algorithm}
The third algorithm, denoted by Alg.3, is based on the $\frac{2}{3}$-approximation algorithm for the maximum weight 2-star packing problem in~\cite{babenko2011new}. It contains the following three steps. 

\medskip
\noindent\textbf{Step~1.} Find $M^*_{n/3}$, the maximum weight matching of size $n/3$ in the graph $G$ used in Alg.2. 
%, using $O(n^3)$ time~\cite{gabow1974implementation,lawler1976combinatorial}.

\noindent\textbf{Step~2.} Find a 2-star packing $S$ in the graph $LG\coloneqq\! G[V(M^*_{n/3})]$ by using the $\frac{2}{3}$-approximation algorithm for the maximum weight 2-star packing problem~\cite{babenko2011new}.

\noindent\textbf{Step~3.} Obtain a 3-path packing $P_3$ of $G$ as follows:
\begin{itemize}
    \item Call each vertex in $V\setminus V(M^*_{n/3})$ as a \emph{residual vertex}.
    \item For each 2-path $xy\in S$, create a 3-path using $xy$ with a residual vertex.
    \item For each 3-path $xyz\in S$, directly select it.
    \item Arbitrarily create vertex-disjoint 3-paths using left residual vertices.
\end{itemize}
\medskip

Note that there are at most $n/3$ 2-paths in $S$, obtained in Step 2 of Alg.3, and hence every 2-path in $S$ can be combined into a 3-path with a residual vertex since there are $n/3$ vertices in $V\setminus V(M^*_{n/3})$.
An illustration of Alg.3 can be seen in Fig.\ref{fig:fig05}.
\begin{figure}[t]
    \centering
    \includegraphics[scale=0.58]{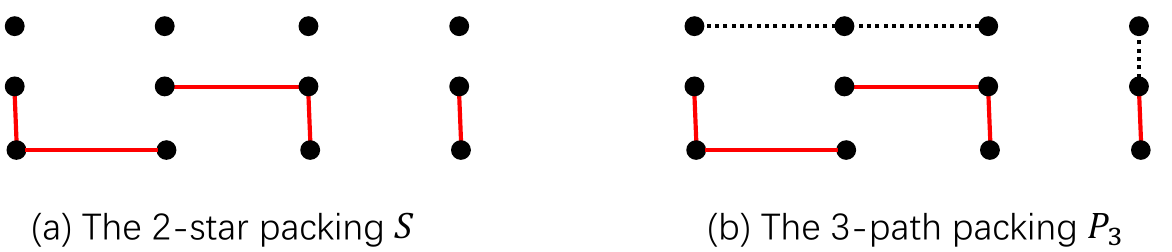}
    \caption{An illustration of Alg.3: In (a), $S$ contains two 3-paths and one edge, represented by red edges; In (b), $S$ is augmented into a 3-path packing in $G$ with four 3-paths.}
    \label{fig:fig05}
\end{figure}

Let $S^*$ be the maximum weight 2-star packing in $LG$, and $A^*$ be the maximum weight 2-feasible arc set in $\overleftrightarrow{LG}$. Since we use the $\frac{2}{3}$-approximation algorithm to obtain $S$ in Step 2 of Alg.3, it holds that $w(S)\geq \frac{2}{3}w(S^*)$. We use a stronger result in their paper. 

\begin{lemma}[\cite{babenko2011new}]\label{lb22}
The 2-star packing $S$ in Step 2 of Alg.3 satisfies that $w(S)\geq \frac{4}{9}w(A^*)$.
\end{lemma}

\begin{lemma}\label{lb23}
It holds that $w(P_3)\geq \frac{4}{9}(2w(Y_5)+2w(Y_6)+w(P^*_7)+w(P^*_8)+\sum_{xyz\in P^*_7\cup P^*_8}\max\{w(xy),w(yz)\})$.  
\end{lemma}
\begin{proof}
By Alg.3, $P_3$ contains all 3-paths in $S$, and for each 2-path in $S$, $P_3$ contains a 3-path containing the 2-path. Hence, we have $w(P_3)\geq w(S)$. Next, we give a lower bound on $w(S)$.

By definitions of $X_i$, $Y_i$, and $P^*_i$, it is easy to see that $S'\coloneqq Y_5\cup Y_6\cup P^*_7\cup P^*_8$ is a set of 2-paths and 3-paths in $G[V(M^*_{n/3})]$. Moreover, we can obtain a 2-feasible arc set $A'$ in $\overleftrightarrow{LG}$ in the following way: for each 2-path $xy\in Y_5\cup Y_6$ obtain two arcs $\{(x,y),(y,x)\}$, and for each 3-path $xyz\in P^*_7\cup P^*_8$ with $w(xy)\geq w(yz)$ obtain three arcs $\{(x,y),(y,z),(y,x)\}$. Hence,
\begin{align*}
w(A')&=\sum_{xy\in Y_5\cup Y_6}2w(xy) + \sum_{xyz\in P^*_7\cup P^*_8}(w(xy)+w(yz)+\max\{w(xy),w(yz)\})\\
&=2w(Y_5)+2w(Y_6)+w(P^*_7)+w(P^*_8)+\sum_{xyz\in P^*_7\cup P^*_8}\max\{w(xy),w(yz)\},
\end{align*}
where the last equality follows from the fact that $X_i$ and $Y_i$ is a decomposition of the edges in $E(P^*_i)$.
Since $w(P_3)\geq w(S)$ and $w(A^*)\geq w(A')$, by Lemma~\ref{lb22}, we have $w(P_3)\geq\frac{4}{9}w(A')\geq\frac{4}{9}(2w(Y_5)+2w(Y_6)+w(P^*_7)+w(P^*_8)+\sum_{xyz\in P^*_7\cup P^*_8}\max\{w(xy),w(yz)\})$. 
\end{proof}

\section{The Trade-Off}
%Recall that $\OPT= w(P^*)$. 
For each $i\in\{1,...,8\}$, we let $\xi_i\coloneqq w(P^*_i)/\OPT$, $\alpha_i\coloneqq w(X_i)/\OPT$, $\beta_i\coloneqq w(Y_i)/\OPT$, and $\gamma_i\coloneqq\sum_{xyz\in P^*_i}\max\{w(xy),w(yz)\}/\OPT$.
Let $\delta\coloneqq w(M^*_{n/2})/\OPT$ and $\pi\coloneqq w(M^*_{n/3})/\OPT$.
For each $i\in\{1,...,5\}$, let $\tau_i\coloneqq w(M_i)/\OPT$. Let $\phi_1\coloneqq w(M'_1)/\OPT$, $\phi_2\coloneqq w(M''_1)/\OPT$, and $\phi_3\coloneqq w(M'''_1)/\OPT$. 
Since $P^*=P^*_1\cup\cdots\cup P^*_8$, $E(P^*_i)=X_i\cup Y_i$, $M^*_{n/3}=M_1\cup\cdots\cup M_5$, and $M_1=M'_1\cup M''_1\cup M'''_1$, we have
\begin{align}
1&=\sum_{i=1}^{8}\xi_i,\label{lp1}\\
\xi_i&=\alpha_i+\beta_i,\quad i\in\{1,...,8\},\\
\pi&=\tau_1+\tau_2+\tau_3+\tau_4+\tau_5,\\
\tau_1&=\phi_1+\phi_2+\phi_3,\\
0&\leq\xi_i,\alpha_i,\beta_i,\gamma_i,\quad i\in\{1,...,8\},\\
0&\leq\delta,\pi,\tau_1,\tau_2,\tau_3,\tau_4,\tau_5,\phi_1,\phi_2,\phi_3.
\end{align}

\begin{comment}
\begin{equation}
\sum_{i=1}^{8}\xi_i=1,
\end{equation}
\begin{equation}
\xi_i=\alpha_i+\beta_i,\quad i\in\{1,...,8\},
\end{equation}
\begin{equation}
\pi=\tau_1+\tau_2+\tau_3+\tau_4+\tau_5,
\end{equation}
\begin{equation}
\tau_1= \phi_1+\phi_2+\phi_3,
\end{equation}
\begin{equation}
\xi_i,\alpha_i,\beta_i,\gamma_i\geq 0,\quad i\in\{1,...,8\},
\end{equation}
\begin{equation}
\delta,\pi,\tau_1,\tau_2,\tau_3,\tau_4,\tau_5,\phi_1,\phi_2,\phi_3\geq 0.
\end{equation}
\end{comment}

By Lemma~\ref{lb8}, we have 
\begin{equation}
w(P_1)\geq\lrA{\frac{2}{3}\delta+\frac{1}{2}-\frac{1}{2}\sum_{i=1}^{8}\gamma_i}\cdot\OPT.
\end{equation}

By Lemmas~\ref{lb11}-\ref{lb16}, we get 
\begin{align}
\alpha_i&\geq \beta_i,\quad i\in\{1,3,4,8\},\\
\tau_3&= \beta_7,\\
\tau_4&=\beta_5,\\
\beta_5&\geq \alpha_5,\\
\delta&\geq\pi+\alpha_1+\beta_2,\\
\delta&\geq\pi-\tau_1-\tau_2+\alpha_1+\gamma_2+\alpha_3+\alpha_4+\alpha_6,\\
\delta&\geq\pi-\tau_1+\alpha_1+\beta_2+\alpha_6,\\
\delta&\geq\pi-\phi_3-\tau_2+\alpha_1+\gamma_2+\alpha_3+\alpha_4.
\end{align}

By Lemma~\ref{lb20}, we have 
\begin{equation}
\begin{split}
w(P_2)&\geq\left(\frac{2}{3}\pi+\frac{1}{12}\phi_1+\frac{1}{6}\phi_2+\frac{1}{4}\phi_3+\frac{1}{6}\tau_2+\frac{1}{4}\tau_3+\frac{1}{3}\tau_4\right.\\
&\quad\quad\left.+\frac{1}{2}\alpha_2+\frac{1}{2}\alpha_3+\frac{1}{2}\alpha_4+\alpha_5+\frac{1}{2}\alpha_7-\frac{1}{4}\beta_7+\frac{1}{4}\xi_6+\frac{1}{4}\xi_8\right)\cdot\OPT.
\end{split}
\end{equation}

By Lemma~\ref{lb23}, we have 
\begin{equation}
w(P_3)\geq \frac{4}{9}\lrA{2\beta_5+2\beta_6+\xi_7+\xi_8+\gamma_7+\gamma_8}\cdot\OPT.
\end{equation}

\begin{lemma}\label{lb24}
It holds that $\pi\geq\sum_{i=1}^{8}\gamma_i$, and $\gamma_i\geq \max\{\alpha_i,\beta_i\}$ for each $i\in\{1,...,8\}$.
\end{lemma}
\begin{proof}
Since $P^*$ is split into the eight pairwise disjoint sets $\{P^*_1,...,P^*_8\}$, by Lemma~\ref{lb1} and the definition of $\gamma_i$, we have 
\begin{align*}
w(M^*_{n/3})&\geq\sum_{xyz\in P^*}\max\{w(xy),w(yz)\}=\sum_{i=1}^{8}\sum_{xyz\in P^*_i}\max\{w(xy),w(yz)\}=\sum_{i=1}^{8}\gamma_i\cdot\OPT.
\end{align*}
Since $w(M^*_{n/3})=\pi\cdot\OPT$ by the definition of $\pi$, we have $\pi\geq \sum_{i=1}^{8}\gamma_i$.

For each $i\in\{1,...,8\}$, by definitions of $X_i$ and $Y_i$, we have $\sum_{xyz\in P^*_i}\max\{w(xy),w(yz)\}=\gamma_i\cdot\OPT\geq \max\{w(X_i),w(Y_i)\}=\max\{\alpha_i,\beta_i\}\cdot\OPT$, and thus $\gamma_i\geq \max\{\alpha_i,\beta_i\}$.
\end{proof}

By Lemma~\ref{lb24}, we have
\begin{align}
\pi&\geq\gamma_1+\gamma_2+\gamma_3+\gamma_4+\gamma_5+\gamma_6+\gamma_7+\gamma_8,\\
\gamma_i&\geq \max\{\alpha_i,\beta_i\},\quad i\in\{1,...,8\}.\label{lp11}
\end{align}

Using (\ref{lp1})-(\ref{lp11}) the approximation ratio $\frac{\max\{w(P_1),\ w(P_2),\ w(P_3)\}}{\OPT}$ can be obtained via solving a linear program. 
Hence, we have the following theorem.

\begin{theorem}\label{main}
For MW3PP, there is a polynomial-time $10/17$-approximation algorithm.
\end{theorem}
\begin{proof}
It is easy to see that our algorithm takes polynomial time.
Using (\ref{lp1})-(\ref{lp11}), the linear program for calculating the approximation ratio can be built as follows. 
(For convenience in deriving the dual, each equality constraint $a = b$ is equivalently rewritten as the pair of inequalities $a \geq b$ and $-a \geq -b$.)
\begin{alignat}{2}
\min\quad &y \nonumber \\
\mbox{s.t.}\quad 
&y\geq\frac{2}{3}\delta+\frac{1}{2}-\frac{1}{2}\sum_{i=1}^{8}\gamma_i, \tag{C1}\\
&y\geq\frac{2}{3}\pi+\frac{1}{12}\phi_1+\frac{1}{6}\phi_2+\frac{1}{4}\phi_3+\frac{1}{6}\tau_2+\frac{1}{4}\tau_3+\frac{1}{3}\tau_4\nonumber\\
&\quad\quad+\frac{1}{2}\alpha_2+\frac{1}{2}\alpha_3+\frac{1}{2}\alpha_4+\alpha_5+\frac{1}{2}\alpha_7-\frac{1}{4}\beta_7+\frac{1}{4}\xi_6+\frac{1}{4}\xi_8,\tag{C2}\\
&y\geq\frac{4}{9}\lrA{2\beta_5+2\beta_6+\xi_7+\xi_8+\gamma_7+\gamma_8}, \tag{C3}\\
&\sum_{i=1}^{8}\xi_i\geq1, \tag{C4}\\
&-\sum_{i=1}^{8}\xi_i\geq -1, \tag{C5}\\
&\xi_i\geq\alpha_i+\beta_i,\quad i\in\{1,...,8\}, \tag{C6-C13}\\
&-\xi_i\geq -(\alpha_i+\beta_i),\quad i\in\{1,...,8\}, \tag{C14-C21}\\
&\pi\geq\tau_1+\tau_2+\tau_3+\tau_4+\tau_5, \tag{C22}\\
&-\pi\geq-(\tau_1+\tau_2+\tau_3+\tau_4+\tau_5), \tag{C23}\\
&\tau_1\geq \phi_1+\phi_2+\phi_3, \tag{C24}\\
&-\tau_1\geq -(\phi_1+\phi_2+\phi_3), \tag{C25}\\
&\alpha_i\geq \beta_i,\quad i\in\{1,3,4,8\},\tag{C26-C29}\\
&\tau_3\geq\beta_7,\tag{C30}\\
&-\tau_3\geq-\beta_7,\tag{C31}\\
&\tau_4\geq\beta_5,\tag{C32}\\
&-\tau_4\geq-\beta_5,\tag{C33}\\
&\beta_5\geq\alpha_5,\tag{C34}\\
&\delta\geq\pi+\alpha_1+\beta_2,\tag{C35}\\
&\delta\geq\pi-\tau_1-\tau_2+\alpha_1+\gamma_2+\alpha_3+\alpha_4+\alpha_6,\tag{C36}\\
&\delta\geq\pi-\tau_1+\alpha_1+\beta_2+\alpha_6,\tag{C37}\\
&\delta\geq\pi-\phi_3-\tau_2+\alpha_1+\gamma_2+\alpha_3+\alpha_4,\tag{C38}\\
&\pi\geq\gamma_1+\gamma_2+\gamma_3+\gamma_4+\gamma_5+\gamma_6+\gamma_7+\gamma_8,\tag{C39}\\
&\gamma_i\geq \alpha_i,\quad i\in\{1,...,8\},\tag{C40-C47}\\
&\gamma_i\geq \beta_i,\quad i\in\{1,...,8\},\tag{C48-C55}\\
&\xi_i,\alpha_i,\beta_i,\gamma_i\geq0,\quad i\in\{1,...,8\},\nonumber\\
&\delta,\pi,\tau_1,\tau_2,\tau_3,\tau_4,\tau_5,\phi_1,\phi_2,\phi_3\geq 0.\nonumber
\end{alignat}
By solving the above linear program, the obtained value is about $0.588235294$, which is very close to $10/17$. 

To prove an approximation ratio of $10/17$, we provide a feasible solution to its dual linear program, which is shown as follows. 

\begin{alignat}{2}
\max\quad &\frac{1}{2}\lambda_1+\lambda_4-\lambda_5 \nonumber \\
\mbox{s.t.}\quad 
&\lambda_1+\lambda_2+\lambda_3\leq 1, \tag{$y$}\\
&\lambda_4-\lambda_5+\lambda_6-\lambda_{14}\leq 0,\tag{$\xi_1$}\\
&\lambda_4-\lambda_5+\lambda_7-\lambda_{15}\leq 0,\tag{$\xi_2$}\\
&\lambda_4-\lambda_5+\lambda_8-\lambda_{16}\leq 0,\tag{$\xi_3$}\\
&\lambda_4-\lambda_5+\lambda_9-\lambda_{17}\leq 0,\tag{$\xi_4$}\\
&\lambda_4-\lambda_5+\lambda_{10}-\lambda_{18}\leq 0,\tag{$\xi_5$}\\
&-\frac{1}{4}\lambda_2+\lambda_4-\lambda_5+\lambda_{11}-\lambda_{19}\leq 0,\tag{$\xi_6$}\\
&-\frac{4}{9}\lambda_3+\lambda_4-\lambda_5+\lambda_{12}-\lambda_{20}\leq 0,\tag{$\xi_7$}\\
&-\frac{1}{4}\lambda_2-\frac{4}{9}\lambda_3+\lambda_4-\lambda_5+\lambda_{13}-\lambda_{21}\leq 0,\tag{$\xi_8$}\\
&-\lambda_6 + \lambda_{14} + \lambda_{26} - \lambda_{35} - \lambda_{36} - \lambda_{37} - \lambda_{38} - \lambda_{40}\leq 0,\tag{$\alpha_1$}\\
&
-\frac{1}{2} \lambda_2 - \lambda_7 + \lambda_{15} - \lambda_{41}\leq 0,\tag{$\alpha_2$}\\
&-\frac{1}{2}\lambda_{2}-\lambda_{8}+\lambda_{16}+\lambda_{27}-\lambda_{36}-\lambda_{38}-\lambda_{42}\leq 0,\tag{$\alpha_3$}\\
&-\frac{1}{2}\lambda_{2}-\lambda_{9}+\lambda_{17}+\lambda_{28}-\lambda_{36}-\lambda_{38}-\lambda_{43}\leq 0,\tag{$\alpha_4$}\\
&-\lambda_{2}-\lambda_{10}+\lambda_{18}-\lambda_{34}-\lambda_{44}\leq 0,\tag{$\alpha_5$}\\
&-\lambda_{11}+\lambda_{19}-\lambda_{36}-\lambda_{37}-\lambda_{45}\leq 0,\tag{$\alpha_6$}\\
&-\frac{1}{2}\lambda_{2}-\lambda_{12}+\lambda_{20}-\lambda_{46}\leq 0,\tag{$\alpha_7$}\\
&-\lambda_{13}+\lambda_{21}+\lambda_{29}-\lambda_{47}\leq 0,\tag{$\alpha_8$}\\
&-\lambda_{6}+\lambda_{14}-\lambda_{26}-\lambda_{48}\leq 0,\tag{$\beta_1$}\\
&-\lambda_{7}+\lambda_{15}-\lambda_{35}-\lambda_{37}-\lambda_{49}\leq 0,\tag{$\beta_2$}\\
&-\lambda_{8}+\lambda_{16}-\lambda_{27}-\lambda_{50}\leq 0,\tag{$\beta_3$}\\
&-\lambda_{9}+\lambda_{17}-\lambda_{28}-\lambda_{51}\leq 0,\tag{$\beta_4$}\\
&-\frac{8}{9}\lambda_{3}-\lambda_{10}+\lambda_{18}-\lambda_{32}+\lambda_{33}+\lambda_{34}-\lambda_{52}\leq 0,\tag{$\beta_5$}\\
&-\frac{8}{9}\lambda_{3}-\lambda_{11}+\lambda_{19}-\lambda_{53}\leq 0,\tag{$\beta_6$}\\
&\frac{1}{4}\lambda_{2}-\lambda_{12}+\lambda_{20}-\lambda_{30}+\lambda_{31}-\lambda_{54}\leq 0,\tag{$\beta_7$}\\
&-\lambda_{13}+\lambda_{21}-\lambda_{29}-\lambda_{55}\leq 0,\tag{$\beta_8$}\\
&\frac{1}{2}\lambda_{1}-\lambda_{39}+\lambda_{40}+\lambda_{48}\leq 0, \tag{$\gamma_1$}\\
&\frac{1}{2}\lambda_{1}-\lambda_{36}-\lambda_{38}-\lambda_{39}+\lambda_{41}+\lambda_{49}\leq 0, \tag{$\gamma_2$}\\
&\frac{1}{2}\lambda_{1}-\lambda_{39}+\lambda_{42}+\lambda_{50}\leq 0, \tag{$\gamma_3$}\\
&\frac{1}{2}\lambda_{1}-\lambda_{39}+\lambda_{43}+\lambda_{51}\leq 0, \tag{$\gamma_4$}\\
&\frac{1}{2}\lambda_{1}-\lambda_{39}+\lambda_{44}+\lambda_{52}\leq 0, \tag{$\gamma_5$}\\
&\frac{1}{2}\lambda_{1}-\lambda_{39}+\lambda_{45}+\lambda_{53}\leq 0, \tag{$\gamma_6$}\\
&\frac{1}{2}\lambda_{1}-\frac{4}{9}\lambda_3-\lambda_{39}+\lambda_{46}+\lambda_{54}\leq 0, \tag{$\gamma_7$}\\
&\frac{1}{2}\lambda_{1}-\frac{4}{9}\lambda_3-\lambda_{39}+\lambda_{47}+\lambda_{55}\leq 0, \tag{$\gamma_8$}\\
&-\frac{2}{3}\lambda_{1}+\lambda_{35}+\lambda_{36}+\lambda_{37}+\lambda_{38}\leq 0, \tag{$\delta$}\\
&-\frac{2}{3}\lambda_{2}+\lambda_{22}-\lambda_{23}-\lambda_{35}-\lambda_{36}-\lambda_{37}-\lambda_{38}+\lambda_{39}\leq 0, \tag{$\pi$}\\
&-\lambda_{22}+\lambda_{23}+\lambda_{24}-\lambda_{25}+\lambda_{36}+\lambda_{37}\leq 0, \tag{$\tau_1$}\\
&-\frac{1}{6}\lambda_2-\lambda_{22}+\lambda_{23}+\lambda_{36}+\lambda_{38}\leq 0, \tag{$\tau_2$}\\
&-\frac{1}{4}\lambda_2-\lambda_{22}+\lambda_{23}+\lambda_{30}-\lambda_{31}\leq 0, \tag{$\tau_3$}\\
&-\frac{1}{3}\lambda_2-\lambda_{22}+\lambda_{23}+\lambda_{32}-\lambda_{33}\leq 0, \tag{$\tau_4$}\\
&-\lambda_{22}+\lambda_{23}\leq 0, \tag{$\tau_5$}\\
&-\frac{1}{12}\lambda_2-\lambda_{24}+\lambda_{25}\leq 0, \tag{$\phi_1$}\\
&-\frac{1}{6}\lambda_2-\lambda_{24}+\lambda_{25}\leq 0, \tag{$\phi_2$}\\
&-\frac{1}{4}\lambda_2-\lambda_{24}+\lambda_{25}+\lambda_{38}\leq 0, \tag{$\phi_3$}\\
&\lambda_i\geq0,\quad i\in\{1,...,55\}.\nonumber
\end{alignat}

There are 55 main constraints in the primal linear program, each corresponding to one of the 55 variables in the dual linear program, denoted by $\lambda_1, \dots, \lambda_{55}$.
Moreover, each constraint in the dual linear program corresponds to a variable in the primal linear program, which is indicated after the constraint.

Consider the following setting: $\lambda_{1}=\frac{24}{51}$, $\lambda_{2}=\frac{24}{51}$, $\lambda_{3}=\frac{3}{51}$, $\lambda_{4}=\frac{18}{51}$, $\lambda_{14}=\frac{18}{51}$, $\lambda_{15}=\frac{18}{51}$, $\lambda_{16}=\frac{18}{51}$, $\lambda_{17}=\frac{18}{51}$, $\lambda_{18}=\frac{18}{51}$, $\lambda_{19}=\frac{12}{51}$, $\lambda_{20}=\frac{50}{153}$, $\lambda_{21}=\frac{32}{153}$, $\lambda_{25}=\frac{2}{51}$, $\lambda_{26}=\frac{18}{51}$, $\lambda_{27}=\frac{18}{51}$, $\lambda_{28}=\frac{18}{51}$, $\lambda_{30}=\frac{6}{51}$, $\lambda_{32}=\frac{8}{51}$, $\lambda_{35}=\frac{10}{51}$, $\lambda_{37}=\frac{2}{51}$, $\lambda_{38}=\frac{4}{51}$, $\lambda_{39}=\frac{32}{51}$, $\lambda_{40}=\frac{20}{51}$, $\lambda_{41}=\frac{6}{51}$, $\lambda_{42}=\frac{20}{51}$, $\lambda_{43}=\frac{20}{51}$, $\lambda_{45}=\frac{10}{51}$, $\lambda_{46}=\frac{14}{153}$, $\lambda_{47}=\frac{32}{153}$, $\lambda_{49}=\frac{6}{51}$, $\lambda_{52}=\frac{22}{153}$, $\lambda_{53}=\frac{28}{153}$, $\lambda_{54}=\frac{50}{153}$, and $\lambda_{55}=\frac{32}{153}$, with $\lambda_i = 0$ for all other $i$. It can be verified that this setting forms a feasible solution to the dual linear program, with an objective value of $\frac{1}{2}\lambda_1+\lambda_4-\lambda_5=\frac{10}{17}$.

Therefore, the dual linear program has a value of least $\frac{10}{17}$, and thus the approximation ratio obtained by the primal linear program is also at least $\frac{10}{17}$.
The details are put in the appendix.
\end{proof}

\begin{remark}
By calculation, we provide the values of the variables in the linear program in one of the worst cases, where we have $\xi_2=\xi_4=\xi_5=\xi_6=\xi_7=0$, $\xi_1=\frac{1}{68}$, $\xi_3=\frac{7}{68}$, $\xi_8=\frac{15}{17}$, $\alpha_i=\beta_i=\gamma_i=\frac{1}{2}\xi_i$ for each $i\in\{1,...,8\}$, $\delta=\frac{69}{136}$, $\pi=\frac{1}{2}$, $\tau_1=\tau_3=\tau_4=0$, $\tau_2=\frac{7}{136}$, $\tau_5=\frac{61}{136}$, and $\phi_1=\phi_2=\phi_3=0$.
\end{remark}

\section{Conclusion}
In this paper, we propose a $\frac{10}{17}$-approximation algorithm for MW3PP, improving the previous claimed approximation ratio of $\frac{7}{12}$. Our algorithm is based on three algorithms. The analysis is based on a newly proposed charging method.
In the future, it would be interesting to extend our methods to design and analyze algorithms for related problems, such as MW3CP~\cite{DBLP:journals/jco/ChenCLWZ21,DBLP:journals/dam/ChenTW09}, MW4PP~\cite{hassin1997approximation}, and MW4CP~\cite{DBLP:journals/tcs/ZhaoX25}.

We remark that for MW4PP, the best-known $3/4$-approximation algorithm~\cite{hassin1997approximation} has remained the state-of-the-art for nearly 30 years.
A similar situation can be observed: in the worst case, the weight of a maximum weight matching of size $n/2$ is roughly the same as that of a maximum weight matching of size $n/4$.
Whether our method can be used to break the $3/4$-approximation barrier remains an open problem.

Moreover, it would be worthwhile to explore better constructions of 2-star packings than the one in Lemma~\ref{lb22}, which is based on~\cite{babenko2011new}. Such improvements may further enhance the approximation ratio for MW3PP.

\bibliographystyle{plain}
\bibliography{main}
% \newpage
\appendix
\section{The Calculation of the Approximation Ratio}

\begin{remark}
The Python code (with Gurobi) for the primal linear program is shown in List~\ref{111111}.\footnote{See also \url{https://github.com/JingyangZhao/MW3PP}.}
\end{remark}

% \nolinenumbers
\nonumber
\begin{lstlisting}[caption={The Python code for solving the linear grogram}, label={111111}]
from gurobipy import *
def mw3pp():

    m = Model('calculate ratio')
    y = m.addVar(vtype=GRB.CONTINUOUS, name="y")
    
    #########################  variables  #########################
    
    xi = m.addVars([1,2,3,4,5,6,7,8], vtype=GRB.CONTINUOUS)
    alpha = m.addVars([1,2,3,4,5,6,7,8], vtype=GRB.CONTINUOUS)
    beta = m.addVars([1,2,3,4,5,6,7,8], vtype=GRB.CONTINUOUS)
    gamma = m.addVars([1,2,3,4,5,6,7,8], vtype=GRB.CONTINUOUS)
    
    delta = m.addVar(vtype=GRB.CONTINUOUS)
    pi = m.addVar(vtype=GRB.CONTINUOUS)
    
    tau = m.addVars([1,2,3,4,5], vtype=GRB.CONTINUOUS)
    phi = m.addVars([1,2,3], vtype=GRB.CONTINUOUS)
    
    #########################  constraints  #########################
    
    m.addConstr( y >= 2/3*delta + 1/2 - 1/2*quicksum(gamma[i] for i in [1,2,3,4,5,6,7,8]) )
    
    m.addConstr( y >= 2/3*pi + 1/12*phi[1] + 1/6*phi[2] + 1/4*phi[3] + 1/6*tau[2] + 1/4*tau[3] + 1/3*tau[4] + 1/2*alpha[2] + 1/2*alpha[3] + 1/2*alpha[4] + alpha[5] + 1/2*alpha[7] - 1/4*beta[7] + 1/4*xi[6] + 1/4*xi[8] )
    
    m.addConstr( y >= 4/9*(2*beta[5] + 2*beta[6] + xi[7] + xi[8] + gamma[7] + gamma[8]) )
    
    m.addConstr( xi.sum('*') >= 1 )
    m.addConstr( -xi.sum('*') >= -1 )
    
    m.addConstrs( xi[i] >= alpha[i] + beta[i] for i in [1,2,3,4,5,6,7,8] )
    m.addConstrs( -xi[i] >= -(alpha[i] + beta[i]) for i in [1,2,3,4,5,6,7,8] )
    
    m.addConstr( pi >= quicksum(tau[i] for i in [1,2,3,4,5]) )
    m.addConstr( -pi >= -quicksum(tau[i] for i in [1,2,3,4,5]) )
    
    m.addConstr( tau[1] >= quicksum(phi[i] for i in [1,2,3]) )
    m.addConstr( -tau[1] >= -quicksum(phi[i] for i in [1,2,3]) )
    
    m.addConstrs( alpha[i] >= beta[i] for i in [1,3,4,8] )
    
    m.addConstr( tau[3] >= beta[7] )
    m.addConstr( -tau[3] >= -beta[7] )
    
    m.addConstr( tau[4] >= beta[5] )
    m.addConstr( -tau[4] >= -beta[5] )
    
    m.addConstr( beta[5] >= alpha[5] )
    
    m.addConstr( delta >= pi + alpha[1] + beta[2] )
    
    m.addConstr( delta >= pi - tau[1] - tau[2] + alpha[1] + gamma[2] + alpha[3] + alpha[4] + alpha[6] )
    
    m.addConstr( delta >= pi - tau[1] + alpha[1] + beta[2] + alpha[6] )
    
    m.addConstr( delta >= pi - phi[3] - tau[2] + alpha[1] + gamma[2] + alpha[3] + alpha[4] )
    
    m.addConstr( pi >= quicksum(gamma[i] for i in [1,2,3,4,5,6,7,8]) )
    
    m.addConstrs( gamma[i] >= alpha[i] for i in [1,2,3,4,5,6,7,8] )
    
    m.addConstrs( gamma[i] >= beta[i] for i in [1,2,3,4,5,6,7,8] )  

    m.addConstrs( xi[i] >= 0 for i in [1,2,3,4,5,6,7,8] )
    m.addConstrs( alpha[i] >= 0 for i in [1,2,3,4,5,6,7,8] )
    m.addConstrs( beta[i] >= 0 for i in [1,2,3,4,5,6,7,8] )
    m.addConstrs( gamma[i] >= 0 for i in [1,2,3,4,5,6,7,8] )
    
    m.addConstr( delta >= 0 )
    m.addConstr( pi >= 0 )
    m.addConstrs( tau[i] >= 0 for i in [1,2,3,4,5] )
    m.addConstrs( phi[i] >= 0 for i in [1,2,3] )
    
    #########################  solve  #########################
    
    m.setObjective(y, GRB.MINIMIZE); m.optimize()
    if m.status == GRB.OPTIMAL: return m.objVal
    else: return -1
mw3pp()
\end{lstlisting}

\begin{remark} 
Recall that at the end of Section~\ref{Second} we mention that, under the constraint that $\sum_{i=1}^{6}w(P^*_i)/\OPT\geq 10^{-4}$, we can achieve an approximation ratio of at least $0.583335$ by making a trade-off between Alg.1 and Alg.2.
It can be verified by deleting the constraint:
\begin{lstlisting}
    m.addConstr( y >= 4/9*(2*beta[5] + 2*beta[6] + xi[7] + xi[8] + gamma[7] + gamma[8]) )
\end{lstlisting}
and adding the following constraint.
\begin{lstlisting}
    m.addConstr( quicksum(xi[i] for i in [1,2,3,4,5,6]) >= 0.0001 )
\end{lstlisting}
\end{remark}

\begin{remark} 
The Python code for verifying the solution to the dual linear program is shown in List~\ref{22222}, where each $\lambda_i$ is scaled to an integer by multiplying by 153, and the coefficients in each constraint are also scaled to integers.
\end{remark}

% \nolinenumbers
\nonumber
\begin{lstlisting}[caption={The Python code for verifying the solution to the dual linear grogram}, label={22222}]
def checkInequality(x, y):
    if x <= y:
        return True
    return False

lambdas = [0, 72, 72, 9, 54, 0, 0, 0, 0, 0, 0, 0, 0, 0, 54, 54, 54, 54, 54, 36, 50, 32, 0, 0, 0, 6, 54, 54, 54, 0, 18, 0, 24, 0, 0, 30, 0, 6, 12, 96, 60, 18, 60, 60, 0, 30, 14, 32, 0, 18, 0, 0, 22, 28, 50, 32]

valid = True;

# Inequality (y)
valid &= checkInequality(lambdas[1] + lambdas[2] + lambdas[3], 153);

# Inequality (\xi_1) to (\xi_8)
valid &= checkInequality(lambdas[4] - lambdas[5] + lambdas[6] - lambdas[14], 0);
valid &= checkInequality(lambdas[4] - lambdas[5] + lambdas[7] - lambdas[15], 0);
valid &= checkInequality(lambdas[4] - lambdas[5] + lambdas[8] - lambdas[16], 0);
valid &= checkInequality(lambdas[4] - lambdas[5] + lambdas[9] - lambdas[17], 0);
valid &= checkInequality(lambdas[4] - lambdas[5] + lambdas[10] - lambdas[18], 0);
valid &= checkInequality(-lambdas[2] + 4 * lambdas[4] - 4 * lambdas[5] + 4 * lambdas[11] - 4 * lambdas[19], 0);
valid &= checkInequality(-4 * lambdas[3] + 9 * lambdas[4] - 9 * lambdas[5] + 9 * lambdas[12] - 9 * lambdas[20], 0);
valid &= checkInequality(-9 * lambdas[2] - 16 * lambdas[3] + 36 * lambdas[4] - 36 * lambdas[5] + 36 * lambdas[13] - 36 * lambdas[21], 0);

# Inequality (\alpha_1) to (\alpha_8)
valid &= checkInequality(-lambdas[6] + lambdas[14] + lambdas[26] - lambdas[35] - lambdas[36] - lambdas[37] - lambdas[38] - lambdas[40], 0);
valid &= checkInequality(-lambdas[2] - 2 * lambdas[7] + 2 * lambdas[15] - 2 * lambdas[41], 0);
valid &= checkInequality(-lambdas[2] - 2 * lambdas[8] + 2 * lambdas[16] + 2 * lambdas[27] - 2 * lambdas[36] - 2 * lambdas[38] - 2 * lambdas[42], 0);
valid &= checkInequality(-lambdas[2] - 2 * lambdas[9] + 2 * lambdas[17] + 2 * lambdas[28] - 2 * lambdas[36] - 2 * lambdas[38] - 2 * lambdas[43], 0);
valid &= checkInequality(-lambdas[2] - lambdas[10] + lambdas[18] - lambdas[34] - lambdas[44], 0);
valid &= checkInequality(-lambdas[11] + lambdas[19] - lambdas[36] - lambdas[37] - lambdas[45], 0);
valid &= checkInequality(-lambdas[2] - 2 * lambdas[12] + 2 * lambdas[20] - 2 * lambdas[46], 0);
valid &= checkInequality(-lambdas[13] + lambdas[21] + lambdas[29] - lambdas[47], 0);

# Inequality (\beta_1) to (\beta_8)
valid &= checkInequality(-lambdas[6] + lambdas[14] - lambdas[26] - lambdas[48], 0);
valid &= checkInequality(-lambdas[7] + lambdas[15] - lambdas[35] - lambdas[37] - lambdas[49], 0);
valid &= checkInequality(-lambdas[8] + lambdas[16] - lambdas[27] - lambdas[50], 0);
valid &= checkInequality(-lambdas[9] + lambdas[17] - lambdas[28] - lambdas[51], 0);
valid &= checkInequality(-8 * lambdas[3] - 9 * lambdas[10] + 9 * lambdas[18] - 9 * lambdas[32] + 9 * lambdas[33] + 9 * lambdas[34] - 9 * lambdas[52], 0);
valid &= checkInequality(-8 * lambdas[3] - 9 * lambdas[11] + 9 * lambdas[19] - 9 * lambdas[53], 0);
valid &= checkInequality(lambdas[2] - 4 * lambdas[12] + 4 * lambdas[20] - 4 * lambdas[30] + 4 * lambdas[31] - 4 * lambdas[54], 0);
valid &= checkInequality(-lambdas[13] + lambdas[21] - lambdas[29] - lambdas[55], 0);

# Inequality (\gamma_1) to (\gamma_8)
valid &= checkInequality(lambdas[1] - 2 * lambdas[39] + 2 * lambdas[40] + 2 * lambdas[48], 0);
valid &= checkInequality(lambdas[1] - 2 * lambdas[36] - 2 * lambdas[38] - 2 * lambdas[39] + 2 * lambdas[41] + 2 * lambdas[49], 0);
valid &= checkInequality(lambdas[1] - 2 * lambdas[39] + 2 * lambdas[42] + 2 * lambdas[50], 0);
valid &= checkInequality(lambdas[1] - 2 * lambdas[39] + 2 * lambdas[43] + 2 * lambdas[51], 0);
valid &= checkInequality(lambdas[1] - 2 * lambdas[39] + 2 * lambdas[44] + 2 * lambdas[52], 0);
valid &= checkInequality(lambdas[1] - 2 * lambdas[39] + 2 * lambdas[45] + 2 * lambdas[53], 0);
valid &= checkInequality(9 * lambdas[1] - 8 * lambdas[3] - 18 * lambdas[39] + 18 * lambdas[46] + 18 * lambdas[54], 0);
valid &= checkInequality(9 * lambdas[1] - 8 * lambdas[3] - 18 * lambdas[39] + 18 * lambdas[47] + 18 * lambdas[55], 0);

# Inequality (\delta)
valid &= checkInequality(-2 * lambdas[1] + 3 * lambdas[35] + 3 * lambdas[36] + 3 * lambdas[37] + 3 * lambdas[38], 0);

# Inequality (\pi)
valid &= checkInequality(-2 * lambdas[2] + 3 * lambdas[22] - 3 * lambdas[23] - 3 * lambdas[35] - 3 * lambdas[36] - 3 * lambdas[37] - 3 * lambdas[38] + 3 * lambdas[39], 0);

# Inequality (\tau_1) to (\tau_5)
valid &= checkInequality(-lambdas[22] + lambdas[23] + lambdas[24] - lambdas[25] + lambdas[36] + lambdas[37], 0);
valid &= checkInequality(-lambdas[2] - 6 * lambdas[22] + 6 * lambdas[23] + 6 * lambdas[36] + 6 * lambdas[38], 0);
valid &= checkInequality(-lambdas[2] - 4 * lambdas[22] + 4 * lambdas[23] + 4 * lambdas[30] - 4 * lambdas[31], 0);
valid &= checkInequality(-lambdas[2] - 3 * lambdas[22] + 3 * lambdas[23] + 3 * lambdas[32] - 3 * lambdas[33], 0);
valid &= checkInequality(-lambdas[22] + lambdas[23], 0);

# Inequality (\phi_1) to (\phi_3)
valid &= checkInequality(-lambdas[2] - 12 * lambdas[24] + 12 * lambdas[25], 0);
valid &= checkInequality(-lambdas[2] - 6 * lambdas[24] + 6 * lambdas[25], 0);
valid &= checkInequality(-lambdas[2] - 4 * lambdas[24] + 4 * lambdas[25] + 4 * lambdas[38], 0);

print(valid)
\end{lstlisting}
\end{document}